\documentclass{aims}
\usepackage{amsmath}
  \usepackage{paralist}
  \usepackage{graphics} 
  \usepackage{epsfig} 
\usepackage{graphicx}  \usepackage{epstopdf}
 \usepackage[colorlinks=true]{hyperref}
\hypersetup{urlcolor=blue, citecolor=red}

\usepackage{mathrsfs, amsthm, amssymb}
\usepackage{bm, hyperref}
\usepackage{cite}
\usepackage{url}
\usepackage{epsfig} 
\usepackage{euscript} 
\usepackage{amsfonts} 
\usepackage{enumerate}

\def\currentvolume{X}
 \def\currentissue{X}
  \def\currentyear{200X}
   \def\currentmonth{XX}
    \def\ppages{X--XX}
     \def\DOI{10.3934/xx.xx.xx.xx}

\newcommand{\ep}{\varepsilon}

\newtheorem{theorem}{Theorem}[section] 
\newtheorem{lemma}[theorem]{Lemma}
\newtheorem{corollary}[theorem]{Corollary}

\newtheorem{proposition}[theorem]{Proposition}

\theoremstyle{definition}
\newtheorem{definition}[theorem]{Definition}

\newtheorem{example}[theorem]{Example}

\newtheorem{remark}[theorem]{Remark}

\newtheorem{notation}[theorem]{Notation}

\newcommand{\C}{\mathbb{C}}
\newcommand{\K}{\mathbb{K}}
\newcommand{\N}{\mathbb{N}}

\newcommand{\R}{\mathbb{R}}
\newcommand{\Z}{\mathbb{Z}}

\newcommand{\bfa}{\mathbf{a}} 

\newcommand{\bfp}{\mathbf{p}}

\newcommand{\bfr}{\mathbf{r}}
\newcommand{\bft}{\mathbf{t}} 
\newcommand{\bfx}{\mathbf{x}}
\newcommand{\bfy}{\mathbf{y}}

\newcommand{\bfPhi}{\mathbf{\Phi}}

\newcommand{\calB}{\mathcal{B}} 
\newcommand{\calD}{\mathcal{D}} 
\newcommand{\calL}{\mathcal{L}}
\newcommand{\calS}{\mathcal{S}} 

\newcommand{\scrB}{\mathscr{B}} 
\newcommand{\scrC}{\mathscr{C}}

\newcommand{\scrM}{\mathscr{M}}

\DeclareMathOperator{\vol}{vol}
\DeclareMathOperator{\Real}{Re}
\DeclareMathOperator{\Imag}{Im}

\newcommand{\res}{\textnormal{res}}

\renewcommand{\labelenumi}{{\normalfont (\roman{enumi})}}

\title[Dimensions, Measurability, Lattice/Nonlattice Dichotomy]{A Survey of Complex Dimensions, Measurability,\\ and the Lattice/Nonlattice Dichotomy}

\author[Dettmers, Giza, Knox, Morales, and Rock]{}



 \email{kmdett@mit.edu}
 \email{rzgiza@gmail.com}
 \email{stinasargent@gmail.com}
 \email{raf.morales01@gmail.com}
 \email{jarock@cpp.edu}

\thanks{The first, second, and fifth authors are supported by NSF grant DMS--1247679.}

\thanks{$^*$ Corresponding author: John A.~Rock}

\begin{document}
\maketitle

\centerline{\scshape Kristin Dettmers}
\medskip
{\footnotesize
 \centerline{Department of Mathematics}
   \centerline{Massachusetts Institute of Technology}
   \centerline{Cambridge, MA 02139, USA}
} 

\medskip

\centerline{\scshape Robert Giza, Rafael Morales, and John A.~Rock$^*$}
\medskip
{\footnotesize
 \centerline{Department of Mathematics and Statistics}
   \centerline{California State Polytechnic University}
   \centerline{Pomona, CA 91768, USA}
}

\medskip

\centerline{\scshape Christina Knox}
\medskip
{\footnotesize
 \centerline{Department of Mathematics}
   \centerline{University of California}
   \centerline{Riverside, CA 92521, USA}
}

\bigskip


\begin{abstract}
The theory of \textit{complex dimensions of fractal strings} developed by Lapidus and van Frankenhuijsen has proven to be a powerful tool for the study of Minkowski measurability of fractal subsets of the real line. In a very general setting, the Minkowski measurability of such sets is characterized by the structure of corresponding complex dimensions. Also, this tool is particularly effective in the setting of self-similar fractal subsets of $\R$ which have been shown to be Minkowski measurable if and only if they are \textit{nonlattice}. This paper features a survey on the pertinent results of Lapidus and van Frankenhuijsen and a preliminary extension of the theory of complex dimensions to subsets of Euclidean space, with an emphasis on self-similar sets that satisfy various separation conditions. This extension is developed in the context of \textit{box-counting measurability}, an analog of Minkowski measurability, which is shown to be characterized by complex dimensions under certain mild conditions. 
\end{abstract}

\section{Introduction}
\label{sec:Introduction}

Let $A$ be a bounded set in Euclidean space $\R^m$. We address the \textit{box-counting content} of $A$ (see Definition \ref{def:BoxCountingMeasurability}), an analog of Minkowski content (see Definition \ref{def:MinkDim}) given by
\begin{align*}
	\scrB(A)&:= \lim_{\varepsilon\to 0^+}\frac{N_B(A,\ep^{-1})}{\ep^{-D}},
\end{align*}
where $D$ is the box-counting dimension of $A$ and $N_B(A,\ep^{-1})$ is the maximum number of disjoint closed balls with centers $\bfa\in A$ and radii $\ep>0$. If $\scrB(A)$ exists in $(0,\infty)$, then $A$ is said to be \textit{box-counting measurable}; see Definition \ref{def:BoxCountingMeasurability}. 

\textit{Self-similar sets} in $\R^m$ provide the key setting for this paper. Let $\bfPhi=\{\varphi_j\}_{j=1}^N$ denote a \textit{self-similar system}, an iterated function system in which each $\varphi_j$ is a contracting similarity acting on $\R^m$. Also, let $F$ denote the unique nonempty compact set satisfying $F=\cup_{j=1}^N\varphi_j(F)$, i.e., $F$ is a self-similar set; see \cite{Hut} as well as Definition \ref{def:ContractingSimilaritySelfSimilarSystemSet}. If the scaling ratio of $\varphi_j$ is denoted by $r_j$ and there are some positive real number $r$ and positive integers $k_j$ such that $r_j=r^{k_j}$, then $\bfPhi$ and $F$ are said to be \textit{lattice}. Otherwise, $\bfPhi$ and $F$ are \textit{nonlattice}; see Definition \ref{def:Lattice}. Although the terminology is not used by Lalley in \cite{Lal88}, the box-counting measurability of self-similar sets is studied therein. Also, some of the results in \cite{Lal88} are used in \cite{Sar}, the new results of which are presented for the first time in Section \ref{sec:BCZFSelfSimilarSets} of this paper.

A conjecture of Lapidus made in the early 1990s claims that, under appropriate conditions, a self-similar set is Minkowski measurable if and only if it is nonlattice. Such a conjecture was proven for subsets of the real line, under assumptions including that the self-similar system in question satisfies the \textit{open set condition} (see Definition \ref{def:OpenSetCondition} as well as \cite{Falc14,Hut,Schief}) and the Minkowski dimension $D$ of the attractor satisfies $0<D<1$ (see the work of Falconer in \cite{Falc95} and Lapidus and van Frankenhuijsen in \cite{LapvF6}, as well as Theorem \ref{thm:LatticeNonlatticeStringMeasurability} in the present text). The conjecture of Lapidus asserts that the statement holds in $\R^m$ for $m\geq 2$ and self-similar sets of Minkowski dimension $D$ where $m-1<D<m$. Gatzouras  was able to prove that if a set $F$ is a nonlattice self-similar set then $F$ is Minkowski measurable; see \cite{Gat00}. The converse remains an open problem from which a substantial amount of active research has stemmed. For instance, see \cite{Kom11,LapPeWi11,LapPeWi13,RatWi13}. An analog of Gatzouras' result in terms of box-counting measurability for self-similar subsets of $\R^m$ with $m\geq 1$ is provided by Theorem \ref{thm:LalleyDichotomy} (a restatement of Theorem 1 in \cite{Lal88}). This theorem is also a key motivation for Corollary \ref{cor:CriterionBCMeasurability} below and other results of this paper. 

Minkowski content has attracted attention due in part to its connection with the (modified) Weyl-Berry conjecture as proven in the context of subsets of $\R$ by Lapidus and Pomerance in \cite{LapPo1}. This result establishes a relationship between the Minkowski content of the boundary of a bounded open set in $\R$ and the spectral asymptotics of the corresponding Laplacian. In turn this led to a reformulation of the Riemann hypothesis as an inverse spectral problem associated with bounded open subsets of $\R$ (see \cite{LapMa}) and the development of the theory of \textit{complex dimensions of fractal strings}; see \cite{LapvF6} as well as Section \ref{sec:TheoryComplexDimensionsFractalStrings} below. 

Also in \cite{LapvF6}, Lapidus and van Frankenhuijsen use complex dimensions to expand upon the lattice/nonlattice dichotomy of certain self-similar sets in $\R$ beyond Minkowski measurability. A \textit{fractal string} is a nonincreasing sequence $\calL=(\ell_j)_{j=1}^\infty$ of positive real numbers which tend to zero; see Definition \ref{def:FractalString}. The set of complex dimensions of $\calL$, denoted by $\calD_\calL(W)$, is the set of poles of a meromorphic extension of the Dirichlet series $\sum_{j=1}^\infty\ell_j^s$ defined for suitable region $W\subseteq\C$; see Definition \ref{def:GeometricZetaFunction}. This Dirichlet series is called the \textit{geometric zeta function} of $\calL$, and its meromorphic extension to any suitable region $W$ is denoted by $\zeta_\calL$. 

Given a bounded open subset of $\R$ denoted by $\Omega$ with infinitely many connected components, the lengths of these components constitute a fractal string. The complex dimensions of this fractal string are used, for instance, to characterize the Minkowski measurability of $\partial\Omega$, the boundary of $\Omega$; see Theorem \ref{thm:CriterionForMinkowskiMeasurability} and \cite{LapvF6}. In the case where $\partial\Omega$ is a self-similar set, the complex dimensions are given by the complex solutions of the corresponding Moran equation (as in \eqref{eqn:Moran} below). That is, 
\begin{align*}
	\calD_\calL	&\subseteq \calS_\bfr:=\left\lbrace s \in \C : 1-\textstyle{\sum_{j=1}^{N}}{r_j}^s = 0 \right\rbrace,
\end{align*}
where the scaling ratios of the self-similar system that define $\partial\Omega$ are given by the \textit{scaling vector} $\bfr=(r_j)_{j=1}^N$. Note that, by Moran's Theorem (Theorem \ref{thm:MoransTheorem}, see also \cite{Falc14,Hut}), the Minkowski dimension of $\partial\Omega$ in this case is the unique positive real number in $\calS_\bfr$. Moreover, deep connections between the structure of the complex dimensions of lattice and nonlattice self-similar sets in $\R$, expanding the lattice/nonlattice dichotomy in this setting, are developed in \cite[Chapter 3]{LapvF6}. In particular, simultaneous Diophantine approximation is used therein to tremendous effect to explore the intricate structure of such complex dimensions.

In the present paper, we make use of many of the results of Lapidus and van Frankenhuijsen in \cite{LapvF6} in terms of the complex dimensions associated with $N_B(A,\cdot)$, the box-counting function of a given bounded set $A$ in $\R^m$ (not just $\R$). In particular, by making use of the \textit{box-counting fractal strings} and developed by Lapidus, \v Zubrini\'c, and the fifth author in \cite{LapRoZu} (see Definition \ref{def:BCFString}), we show  that the box-counting measurability of $A$ is characterized by the structure of the corresponding complex dimensions. This result is presented in Corollary \ref{cor:CriterionBCMeasurability} as an analog of \cite[Theorem 8.15]{LapvF6} (which appears as Theorem \ref{thm:CriterionForMinkowskiMeasurability} below). Additionally, we show that under mild conditions the box-counting content is given by
\begin{align*}
	\scrB(A)	&=\frac{\res(\zeta_\calL(s);D)}{D},
\end{align*}
where $D$ is the box-counting dimension of $A$ (see Definition \ref{def:BoxDim}), $\calL$ is the box-counting fractal string constructed using $N_B(A,\cdot)$, and $\res(\zeta_\calL(s);\omega)$ is the residue of $\zeta_\calL$ at $\omega\in W$ for a suitably defined $W$. Also, in terms of self-similar sets in $\R^m$ that are either \textit{strongly separated} or satisfy either the \textit{open set condition} (see Definitions \ref{def:deltadisjoint} or \ref{def:OpenSetCondition}, respectively), we show that the corresponding complex dimensions are often given by elements $\calS_\bfr$; see Corollaries \ref{cor:BCcdim} and \ref{cor:SOSCComplexDimensions}.

Various attempts to extend the theory of complex dimensions to sets in $\R^m$ have been made. For instance, the approach taken in \cite{LapPeWi11,LapPeWi13} involves similarly defined complex dimensions and the Minkowski measurability and tilings of self-similar sets that satisfy the open set condition along with certain a nontriviality condition, but this approach does not extend to other types of bounded sets in $\R^m$. In \cite{LapRaZu14}, a theory of complex dimensions is developed in the context of \textit{distance} and \textit{tube zeta functions} associated with arbitrary bounded sets in $\R^m$. However, as of the writing of this paper, the results presented therein have not been used to provide an extension of lattice/nonlattice dichotomy to self-similar sets in $\R^m$.
 
The structure of the paper is as follows. In Section \ref{sec:Preliminaries}, we summarize many of the classical results on Minkowski/box-counting dimension and self-similar sets, including a discussion of simultaneous Diophantine approximation and its connection to the lattice/nonlattice dichotomy. In Section \ref{sec:TheoryComplexDimensionsFractalStrings}, we summarize the results on fractal strings and complex dimensions of \cite{LapvF6} that motivate and are used to prove the new results presented in later sections. In Sections \ref{sec:BCZFSelfSimilarSets}, the new results of master's thesis \cite{Sar} pertaining to box-counting fractal strings of self-similar sets in $\R^m$ under various separation conditions are presented along with a couple of key examples. Section \ref{sec:RelatedResultsAndFutureWork} concludes the paper with a discussion of a few results from the master's thesis \cite{Morales}, an application of which includes the criterion for box-counting measurability presented in Corollary \ref{cor:CriterionBCMeasurability}.

\newpage
\section{Preliminaries}
\label{sec:Preliminaries}

\subsection{Dimensions and contents}
\label{sec:DimensionsAndContents}

The following notation and terminology are well-known in the literature on fractal geometry.

\begin{notation}
\label{not:D,v,O}
Let $A\subseteq \R^m$ and $\bfx\in\R^m$.  Let $d(\bfx,A)$ denote the distance between $\bfx$ and $A$ given by $d(\bfx,A):=\inf\{||\bfx-\bfa||_m: \bfa\in A\}$, where $||\cdot||_m$ denotes the usual $m$-dimensional Euclidean norm. The notation $d_m$ denotes the $m$-dimensional Euclidean metric. For $\ep>0$, the \textit{open} $\ep$-\textit{neighborhood} of $A$, denoted by $A_{\ep}$, is the set of points in $\R^m$ within $\ep$ of $A$ given by $A_{\ep}:= \{ \bfx\in \R^m : d(\bfx,A)< \ep \}$. Also, $m$-dimensional Lebesgue measure is denoted by $\vol^m$.
\end{notation}

\begin{definition}
\label{def:MinkDim}
The \textit{upper} and \textit{lower Minkowski dimensions} of a bounded set $A\subseteq \R^m$ are respectively defined by
	\begin{align*}
	\overline{\dim}_MA	&:=m-\liminf_{\ep\to 0^+}\frac{\log\vol^m(A_\ep)}{\log\ep}, \quad \textnormal{and} \quad
	\underline{\dim}_MA	:=m-\limsup_{\ep\to 0^+}\frac{\log\vol^m(A_\ep)}{\log\ep}.
	\end{align*}
When $\overline{\dim}_MA=\underline{\dim}_MA $, the corresponding limit exists and the common value, denoted by $\dim_MA$, is called the \textit{Minkowski dimension} of $A$. In the case where $D_M=\dim_M A$ exists, the \textit{upper} and \textit{lower Minkowski contents} of $A$ are respectively defined by
	\begin{align*}
	\scrM^*(A) &:= \limsup_{\ep \to 0^+} \frac{\vol^m(A_\ep)}{\ep^{m-D_M}}, \qquad \textnormal{and} \qquad
	\scrM_*(A) := \liminf_{\ep \to 0^+} \frac{\vol^m(A_\ep)}{\ep^{m-D_M}}.
	\end{align*}
If $\scrM^*(A)=\scrM_*(A)$, the corresponding limit exists and the common value, denoted by $\scrM(A)$,  is called the \textit{Minkowski content} of $A$. If  $0<\scrM_*(A)=\scrM^*(A)<\infty$, then $A$ is said to be \textit{Minkowski measurable}.
\end{definition}
 
A well-known equivalent formulation of Minkowski dimension is \textit{box-counting dimension}. As noted in \cite{Falc14} and elsewhere, the types of ``boxes'' used in the computation of box-counting dimension can vary. However, in this paper (as in \cite{LapRoZu,Sar}), only the specific notion of box-counting function defined below is considered.

\begin{definition}
\label{def:BoxCountingFunction}
Let $A$ be a bounded subset of $\R^m$. The \textit{box-counting function} of $A$ is the function $N_B(A,\cdot):(0,\infty)\rightarrow\N\cup\{0\}$ where $N_B(A,x)$ denotes the maximum number of disjoint closed balls 
with centers $\bfa\in A$ and radii $x^{-1}>0$.
\end{definition}

\begin{remark}
\label{rmk:NotationConvention}
Note that for $\ep>0$, $N_B(A,\ep^{-1})$ denotes the maximum number of disjoint closed balls 
with centers $\bfa\in A$ and radii $\ep$. Although it may seem unnatural to define the box-counting function in terms of $x=\ep^{-1}$, this notation is used throughout \cite{LapRoZu} and \cite{Sar} in order to make use of the results of \cite{LapvF6}, as done in Section \ref{sec:BCZFSelfSimilarSets} below. Hence, the convention is adopted here as well. 
\end{remark}

\begin{definition}
\label{def:BoxDim}
For a bounded set $A\subseteq\R^m$, the \textit{lower} and \textit{upper box-counting dimensions} of $A$, denoted by $\underline{\dim}_BA$ and $\overline{\dim}_BA$, respectively, are given by 
	\begin{align*}
	\underline{\dim}_BA	&:=\liminf_{x\rightarrow\infty}\frac{\log N_B(A,x)}{\log x} \qquad \textnormal{and} \qquad
	\overline{\dim}_BA	:=\limsup_{x\rightarrow\infty}\frac{\log N_B(A,x)}{\log x}.
	\end{align*}
When $\underline{\dim}_BA=\overline{\dim}_BA$, the corresponding limit exists and the common value, denoted by $\dim_BA$, is called the \textit{box-counting dimension} of $A$.
\end{definition}

The following theorem is a classic result in dimension theory.

\begin{theorem}
\label{thm:BoxCountingEqualsMinkowski}
Let $A\subseteq\R^m$. Then, when either exists, $\dim_MA=\dim_BA$.
\end{theorem}

The terminology \textit{box-counting content} does not seem to be used elsewhere in the literature. However, the concept (see Definition \ref{def:BoxCountingMeasurability}) is a key component of \cite{Lal88} and allows for content to be studied in the context of box-counting functions. Such an approach is central to \cite{Sar}, the new results of which are presented and expanded upon in Section \ref{sec:BCZFSelfSimilarSets} of this paper. Note that, at least intuitively, for a bounded set $A\subseteq\R^m$ we have for some positive  constant $c$ that
\begin{align*}
	\vol^m(A_\ep)	&\approx c\ep^m N_B(A,\ep^{-1}),
\end{align*}
which is to say that the volume of the open $\ep$-neighborhood of a set $A$ is approximately equal to the product of the maximum number of disjoint closed balls of radius $\ep$ with centers in $A$ and the ``size'' of each ball given by $c\ep^m$.

\begin{definition}
\label{def:BoxCountingMeasurability}
Let $A\subseteq \R^m$ be such that $D_B:=\dim_BA$ exists.  The \emph{upper} and \emph{lower box-counting contents} of $A$ are respectively defined by 
\begin{align*}
	\scrB^*(A)&:= \limsup_{x\to\infty}\frac{N_B(A,x)}{x^{D_B}} \quad \textnormal{and} \quad
	\scrB_*(A):= \liminf_{x\to\infty}\frac{N_B(A,x)}{x^{D_B}}.
\end{align*}
If $\scrB^*(A)=\scrB_*(A)$, the corresponding limit exists and the common value, denoted by $\scrB(A)$, is called the \textit{box-counting content} of $A$. If $0<\scrB_*(A)=\scrB^*(A)<\infty$, then $A$ is said to be \emph{box-counting measurable}.
\end{definition}

\begin{example}
\label{eg:UnitIntervalMeasurable}
For the closed unit interval $[0,1]$, it is readily shown  that $\dim_B[0,1]=\dim_M[0,1]=1$ and that $[0,1]$ is Minkowski measurable with $\scrM([0,1])=1$. It is also true that $[0,1]$ is box-counting measurable. In this case, $N_B([0,1],x)=[x/2]+1$ where $[y]$ denotes the integer part of $y\in\R$ and, hence, $\scrB([0,1])=1/2$. 
\end{example}

\subsection{The Hausdorff metric and self-similar sets}
\label{sec:Self-SimilarSets}
  
Self-similar sets can be constructed as fixed points of self-similar systems, as described in the next two definitions. Let $\K$ denote the set of nonempty compact subsets of a Euclidean space $\R^m$ equipped with $d_m$.

\begin{definition}
\label{def:HausdorffMetric}
Let $A,B\in\K$. The \textit{Hausdorff metric} $d_H$ is defined by
\begin{align*}
	d_H(A,B)	&:=\inf\{\varepsilon > 0: A \subseteq [B]_{\varepsilon} \text{ and } B \subseteq [A]_{\varepsilon}\},
\end{align*}
where $[A]_{\varepsilon}$ denotes the closed $\varepsilon$-neighborhood of $A$.
\end{definition}

It is a well-known fact that $\K$ equipped with the Hausdorff metric is a complete metric space (see \cite[2.10.21]{Federer}). Thus, Banach's Fixed-Point Theorem applies and establishes the existence of a unique \textit{attractor} for $\bfPhi$, see \cite[\S 10.3]{FitzRoyd}.

\begin{definition}
\label{def:ContractingSimilaritySelfSimilarSystemSet}
A function $\varphi:\R^m \rightarrow \R^m$ is a \textit{contracting similarity} on $\R^m$ if for all $\bfx,\bfy \in \R^m$ and some $0< r<1$, called the \textit{scaling ratio} of $\varphi$, we have
\begin{align*}
	d_m(\varphi(\bfx),\varphi(\bfy))=rd_m(\bfx,\bfy).
\end{align*}
A \textit{self-similar system} on $\R^m$ is a finite collection $\bfPhi=\{\varphi_j\}_{j=1}^N$ of $N\geq 2$ contracting similarities on $\R^m$. By a mild abuse of notation, a self-similar system $\bfPhi$ may be interpreted as a map $\bfPhi:\K \rightarrow \K$ defined by $\bfPhi(\cdot):=\bigcup^N_{j=1} \varphi_j(\cdot)$. The \textit{scaling vector} of $\bfPhi$ is given by $\bfr=(r_j)^N_{j=1}$ where, for each $j=1,2,\ldots,N$, the contracting similarity $\varphi_j$ has scaling ratio $r_j$. The \textit{attractor} of $\bfPhi$ is the unique nonempty compact set $F\in\K$ such that $\bfPhi(F)=F$. Also, a \textit{self-similar set} is the attractor of a self-similar system.
\end{definition}

Contracting similarities have the following characterization (see Proposition 2.3.1 of \cite{Hut}) which is used in the proof of Theorems \ref{thm:LatticeApproximation} and \ref{thm:LatticeApproximationAdditional} below.

\begin{proposition}
\label{prop:SelfSimilarCharacterization}
A function $\varphi$ is a contracting similarity on $\R^m$ with scaling ratio $r$ if and only if $\varphi(\cdot)=rQ(\cdot)+\mathbf{t}$, where $Q$ is an orthonormal transformation, $\bft$ is a fixed translation vector, and $0<r<1$.
\end{proposition}

The following examples are revisited throughout the paper. 

\begin{example}
\label{eg:CantorSetSelfSimilar}
The Cantor set $C$ is the attractor of the self-similar system $\bfPhi_C$ on $\R$ given by the following pair of contracting similarities: 
\begin{align*}
	\varphi_1(x)&= \frac{x}{3}, \qquad \textnormal{and} \qquad \varphi_2(x) = \frac{x}{3}+\frac{2}{3}.
\end{align*}
The scaling vector of $\bfPhi_C$ is given by $\bfr_C=(1/3,1/3)$. 
\end{example}

\begin{example}
\label{eg:GoldenStringSystem}
Consider the self-similar system $\bfPhi_\phi$ on $\R$ given by
\begin{align*}
	\varphi_1(x) &= \frac{x}{2}+\frac{1}{2}, \qquad \textnormal{and} \qquad
	\varphi_2(x) = \frac{x}{2^\phi},
\end{align*}
where $\phi=(1+\sqrt{5})/2$ (the Golden Ratio). The attractor of this self-similar system is a Cantor-like set denoted by $A_\phi$ and its scaling vector is given by $\bfr_\phi=(1/2,1/2^\phi)$. 
\end{example}

\begin{example}
\label{eg:SierpinskiGasket}
The Sierpi\'{n}ski gasket $S_G$ is the attractor of the self-similar system $\bfPhi_S$ on $\R^2$ given by the following three contracting similarities: 
\begin{align*}
	\varphi_1(\bfx)	&=\frac{1}{2}\bfx, \quad \varphi_2(\bfx)	=\frac{1}{2}\bfx+\left(\frac{1}{2},0\right), \quad \textnormal{and} \quad \varphi_3(\bfx)	=\frac{1}{2}\bfx+\left(\frac{1}{4},\frac{\sqrt{3}}{4}\right).
\end{align*}
The scaling vector of $\bfPhi_S$ is $\bfr_S=(1/2,1/2,1/2)$. See Figure \ref{fig:figure}.
\end{example}

\begin{example}
\label{eg:4byQuarter}
Consider the self-similar set given by the attractor $F_1$ of the self-similar system $\bfPhi_1=\{\varphi_j\}_{j=1}^4$ on $\R^2$ given by
\begin{align*}
\varphi_1(\bfx)	&=\frac{1}{4}\bfx, \quad \varphi_2(\bfx)	=\frac{1}{4}\bfx+\left(\frac{3}{4},0\right), \quad
\varphi_3(\bfx)	=\frac{1}{4}\bfx+\left(\frac{3}{4},\frac{3}{4}\right), \quad \textnormal{and}\\ \varphi_4(\bfx)	&=\frac{1}{4}\bfx+\left(0,\frac{3}{4}\right).
\end{align*} The scaling vector of $\bfPhi_1$ is $\bfr_1=(1/4,1/4,1/4,1/4)$. See Figure \ref{fig:figure}.
\end{example}

\begin{figure}
\begin{center}
\includegraphics[scale=.25]{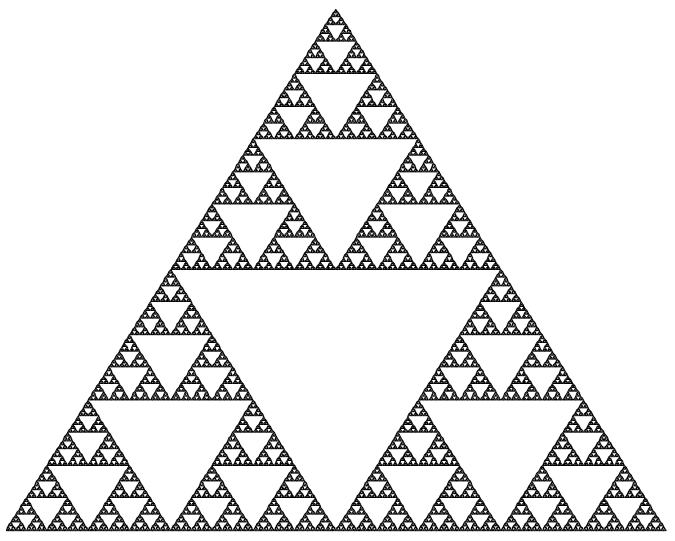}  \hspace{1cm}
\includegraphics[scale=.40]{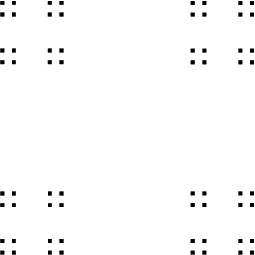}
\end{center}
\caption{On the left, an approximation of the Sierpi\'{n}ski gasket $S_G$ discussed in Example \ref{eg:SierpinskiGasket}. On the right, an approximation of the self-similar set $F_1$ discussed in Example \ref{eg:4byQuarter}.}
\label{fig:figure}
\end{figure}

\subsection{Separation conditions and Moran's Theorem}
\label{sec:OSCMoransTheorem}
The open set condition and the strong open set condition are common in the literature on self-similar sets. For instance, see \cite{Falc14,Gat00,Hut,Kom11,Lal88,LapPeWi11,LapPeWi13,LapvF6,Sar,{Schief}}. 

\begin{definition}
\label{def:OpenSetCondition}
A self-similar system $\bfPhi=\{\varphi_j\}_{j=1}^N$ is said to satisfy the \textit{open set condition}  if there exists a nonempty open set $V \subseteq \R^m$ such that $\varphi_j(V)\subseteq V$ for each $j$ and the images $\varphi_j(V)$ are pairwise disjoint. If, in addition, $V\cap F \neq \emptyset$ where $F$ is the attractor of $\bfPhi$, then $\bfPhi$ is said to satisfy the \textit{strong open set condition}.
\end{definition}

The open set condition implies the strong open set condition, as shown in \cite{Schief}. 

\begin{theorem}
\label{thm:OSCimpliesSOSC}
A self-similar system satisfies the open set condition if and only if it satisfies the strong open set condition.
\end{theorem}

Thus, it is sufficient to assume the open set condition holds in order to apply a result of Lalley (which uses the strong open set condition, see \cite{Lal88,Lal89}), allowing for the result obtained in Proposition \ref{prop:SOSCerror}. 

\begin{example}
\label{eg:OSCAndSOSC}
Each of the self-similar systems in Examples \ref{eg:CantorSetSelfSimilar}, \ref{eg:GoldenStringSystem}, \ref{eg:SierpinskiGasket}, and \ref{eg:4byQuarter} satisfy the (strong) open set condition. 
Note, however, that an open set may suffice for the open set condition but not the \textit{strong} open set condition. 
\end{example}

An even stronger separation condition is used to generate some of the results presented in Sections \ref{sec:BCZFSelfSimilarSets} and \ref{sec:RelatedResultsAndFutureWork}.

\begin{definition}
\label{def:deltadisjoint}
Let $\bfPhi=\{\varphi_j\}_{j=1}^N$ be a self-similar system with attractor $F$. Then $\bfPhi$ and $F$ are said to be \textit{strongly separated} if the images of $F$ under the contracting similarities $\varphi_j$ are pairwise disjoint. If, in addition, 
\begin{align*}
	\delta &:= \sup\{\alpha: d(\bfx,\bfy)>\alpha \hspace{5pt} \forall \bfx\in \varphi_j(F), \bfy\in \varphi_k(F), j \neq k, j,k=1,\ldots,N\}
\end{align*}
is positive and finite, then $\bfPhi$ and $F$ are said to be $\delta$-\textit{disjoint}.
\end{definition}

\begin{remark}
Note that strongly separated self-similar systems satisfy the (strong) open set condition and their attractors are always totally disconnected. Also, $\bfPhi_C$, $\bfPhi_\phi$, and $\bfPhi_1$ are strongly separated, but $\bfPhi_S$ is not.
\end{remark}

\begin{definition}
\label{def:SimilarityDimension}
Let $\bfPhi$ be a self-similar system that satisfies the open set condition with attractor $F$ and scaling vector $(r_j)_{j=1}^N$ where $N \geq 2$. The \textit{similarity dimension} of $F$, denoted by $D_\bfr$, is the unique nonnegative real solution of the \textit{Moran equation}
\begin{align}
\label{eqn:Moran}
	\sum_{j=1}^N r_j^\sigma	&= 1, \quad \sigma \in \R.
\end{align}
\end{definition}

\begin{theorem}[Moran's Theorem]
\label{thm:MoransTheorem}
Let $\bfPhi$ be a self-similar system that satisfies the open set condition with attractor $F$ and scaling vector $\bfr=(r_j)_{j=1}^N$ where $N \geq 2$. Then $\dim_MF=\dim_BF=D_\bfr$.
\end{theorem}

\begin{remark}
\label{rmk:ComplexMoranSolutions}
Given a scaling vector $\bfr=(r_j)_{j=1}^N$, we are also interested in the set of complex solutions of \eqref{eqn:Moran} denoted by
\begin{align}
\label{eqn:ComplexMoranSolutions}
	\calS_\bfr	:=\left\{\omega\in\C : \textstyle{\sum_{j=1}^N}r_j^\omega=1\right\}.
\end{align}
In particular, these values provide the \textit{complex dimensions} associated with many of the self-similar sets studied throughout this paper. 
\end{remark}

\begin{example}
\label{eg:SimilarityDimensions}
By Theorem  \ref{thm:MoransTheorem}, the self-similar sets $C$, $A_\phi$, $S_G$ and $F_1$ in Examples \ref{eg:CantorSetSelfSimilar}, \ref{eg:GoldenStringSystem}, \ref{eg:SierpinskiGasket}, and \ref{eg:4byQuarter} have the following box-counting dimensions:
\begin{align*}
	\dim_BC	&=\log_32, \quad \dim_BA_\phi=\log_\phi 2, \quad \dim_BS_G=\log_23, \quad \textnormal{and } \dim_BF_1=1.
\end{align*}
In particular, the box-counting dimension of $A_\phi$ is given by the unique positive real number $D_\phi=\dim_BA_\phi=\log_\phi 2$ that satisfies the following equation:
\begin{align*}
	1/2^{D_\phi}+(1/2^{\phi})^{D_\phi}	&= 1, 
\end{align*}
where $\phi$ is the Golden Ratio. The approximation $D_\phi\approx .77921$ follows from Lemma 3.29 of \cite{LapvF6} and serves as an application of the simultaneous Diophantine approximation provided by Lemma 3.16 of \cite{LapvF6}. See Section \ref{sec:SimultaneousDiophantineApproximation} below.
\end{example}

\begin{remark}
\label{rmk:LatticeAndOtherApproximation}
The subset of $\K$ comprising finite sets of vectors with rational components  is countable and dense in $\K$ (which shows that $\K$ is separable). Hence, in the Hausdorff metric, any $A\in\K$ can be approximated as closely as one would like by a finite set. However, such an approximation \textit{does not} correspond to a satisfactory notion of approximation for the \textit{dimensions} of the corresponding sets. Indeed, finite sets are 0-dimensional, but a compact subset of $\R^m$ can have Minkowski dimension anywhere in $[0,m]$.

As discussed in Theorems \ref{thm:LatticeApproximation} and \ref{thm:LatticeApproximationAdditional}, the approximation of a \textit{nonlattice} self-similar set by a sequence of \textit{lattice} self-similar sets in the Hausdorff metric, as well as the convergence of their \textit{box-counting complex dimensions} in the sense described in Remark \ref{rmk:LatticeRootsApprox}, follow from the simultaneous Diophantine approximation of the corresponding scaling vectors described in the next section. 
\end{remark}


\subsection{Lattice/nonlattice dichotomy and measurability}
\label{sec:LatticeNonlatticeDichotomyMeasurability}
Much of the remainder this paper focuses on the \textit{lattice/nonlattice dichotomy} for self-similar sets in Euclidean space, an introduction to which is provided in this section. For the purposes of this paper, the dichotomy is  studied in the context of (inner) Minkowski dimension on the real line in Section \ref{sec:Lattice/NonlatticeDichotomyOfSelfSimilarStrings} and in the context of box-counting dimension in Sections \ref{sec:BCZFSelfSimilarSets} and \ref{sec:RelatedResultsAndFutureWork}). 

\begin{definition}
\label{def:Lattice}
Let $\bfPhi$ be a self-similar system with scaling vector $\bfr=(r_j)^N_{j=1}$ and attractor $F$. Then $\bfPhi$, its scaling vector $\bfr$, and its attractor $F$ are said to be \textit{lattice} if there exist $0<r<1$ and $N$ positive integers $k_j$ such that $\gcd(k_1,\ldots,k_N)=1$ and $r_j=r^{k_j}$ for $j=1,2,\ldots,N$. Otherwise $\bfPhi$, $\bfr$, and $F$ are said to be \textit{nonlattice}. 
\end{definition}

\begin{example}
\label{eg:ExamplesLatticeNonlattice} 
The self-similar systems $\bfPhi_C$, $\bfPhi_S$, and $\bfPhi_1$ in Examples \ref{eg:CantorSetSelfSimilar}, \ref{eg:SierpinskiGasket}, and \ref{eg:4byQuarter}, respectively, are lattice. On the other hand, the self-similar system $\bfPhi_\phi$ in Example \ref{eg:GoldenStringSystem} is nonlattice since $\bfr_\phi=(1/2,1/2^\phi)$ and $\phi=(1+\sqrt{5})/2$ is irrational. See Example \ref{eg:GoldenString} below for a discussion on the componentwise lattice approximation of $\bfr_\phi$ and see Remark \ref{rmk:LatticeRootsApprox} for a discussion regarding the ``quasiperiodic'' structure of the complex dimensions of the \textit{Golden string} $\Omega_\phi=[0,1]\backslash A_\phi$.
\end{example}

In the sequel, we write $f(x)\sim g(x)$ as $x\rightarrow a$ when $\lim_{x\rightarrow a} f(x)/g(x)=1$.

\begin{theorem}
\label{thm:LalleyDichotomy}
Let $\bfPhi$ be a self-similar system that satisfies the open set condition with attractor $F$ and scaling vector $\bfr$. Let $D:=\dim_BF=D_\bfr$.
\begin{itemize}
	\item[(a)][Nonlattice] If the additive group generated by $\{\log{r_j}\}$ is dense in $\R$, then there exists $C>0$ such that
	\begin{align*}
		N_B(F,x)	&\sim Cx^D, \quad \textnormal{as } x\to\infty. 
	\end{align*}
	\item[(b)][Lattice] If the additive group generated by $\{\log{r_j}\}$ is $h\Z$ for some $h>0$, then for each $\beta\in[0,h)$ there exists $C_\beta>0$ such that
	\begin{align*}
		N_B(F,e^{nh-\beta})	&\sim C_\beta e^{D(nh-\beta)}, \quad \textnormal{as } n\to\infty. 
	\end{align*}
\end{itemize} 
\end{theorem}

\begin{remark}
\label{rmk:LalleyAsContent}
Part (a) of Theorem \ref{thm:LalleyDichotomy} says that nonlattice sets are box-counting measurable (see Definition \ref{def:BoxCountingMeasurability}) and part (b) says that lattice sets exhibit a log-periodic structure. Additionally, if for some lattice set $F$ there exist $\beta_1,\beta_2\in[0,h)$ such that $C_{\beta_1}\neq C_{\beta_2}$, then $F$ is not box-counting measurable and $F$ is said to exhibit \textit{geometric oscillations of order} $D=\dim_BF=D_\bfr$. This is the case for the 1-dimensional set $F_1$ and the Sierpi\'{n}ski gasket $S_G$ with $D=1$ and $D=\log_23$, respectively. See Examples \ref{eg:SierpinskiGasket}, \ref{eg:4byQuarter}, \ref{eg:BCZF4byQuarter}, and \ref{eg:BCZFGasket}.
\end{remark}

\begin{example}
\label{eg:UnitIntervalLatticeNonlattice}
The closed unit interval $[0,1]$ is a self-similar set generated by both lattice and nonlattice self-similar systems. For instance, a corresponding lattice scaling vector is $(1/2,1/2)$ and a corresponding nonlattice scaling vector is $(1/2,1/3,1/6)$. As noted in Example \ref{eg:UnitIntervalMeasurable}, $[0,1]$ is both Minkowski measurable and box-counting measurable. 
Part (a) of Theorem \ref{thm:LalleyDichotomy} also implies that $[0,1]$ is box-counting measurable. Part (b) of Theorem \ref{thm:LalleyDichotomy} applies as well, but in this case $C_\beta=\scrB([0,1])=1/2$ for all $\beta\in[0,h)$.
\end{example}

\subsection{Simultaneous Diophantine approximation}
\label{sec:SimultaneousDiophantineApproximation}

The simultaneous Diophantine approximation provided by Lemma 3.16 of \cite{LapvF6}
says that if at least one of a finite set of real numbers $\alpha_1,\ldots,\alpha_N$ is irrational, then these numbers can be approximated by rational numbers \textit{with a common denominator}. Thus, one can find integers $q$ such that for each $j=1,\ldots,N$, the product $q\alpha_j$ is within a small distance to the nearest integer. This leads to Lemma \ref{lem:SequenceOfLatticeScalingVectors}. Note that this simultaneous Diophantine approximation also leads to the convergence of a sequence of lattice sets to a given nonlattice set in the Hausdorff metric  (as in Theorem \ref{thm:LatticeApproximation}) as well as the convergence of corresponding \textit{complex dimensions} (as in Theorem \ref{thm:LatticeApproximationAdditional} below and described heuristically in Remark \ref{rmk:LatticeRootsApprox}).

\begin{lemma}
\label{lem:SequenceOfLatticeScalingVectors}
Let $\bfr=(r_j)^N_{j=1}$ be a nonlattice scaling vector. Then there exists a sequence ${({\bfr}_M)}^\infty_{M=1}$ of lattice scaling vectors, each of order $N$, such that ${\bfr}_M \rightarrow \bfr$ componentwise as $M \rightarrow \infty$.
\end{lemma}

\begin{example}
\label{eg:GoldenString}
A simple yet illustrative example of the convergence described in Lemma \ref{lem:SequenceOfLatticeScalingVectors} stems from the nonlattice set $A_\phi$ in Examples \ref{eg:GoldenStringSystem} and \ref{eg:ExamplesLatticeNonlattice} (see \cite{LapvF6}). It is well known that $\phi$ is approximated by ratios of consecutive terms of the Fibonacci sequence $(f_M)_{M=0}^\infty$ (using the convention $f_0=0$ and $f_1=1$), i.e. $f_{M+1}/f_M\to \phi$ as $M\to\infty$. Thus, $\bfr_\phi$ is approximated by 
\begin{align*}
	\bfr_M	&= \displaystyle{\left(\frac{1}{2},\frac{1}{2^{f_{M+1}/f_M}}\right)},
\end{align*}
where $M\geq 1$. Then with $r=1/2^{1/f_M}$, $k_1 = f_M$, and $k_2 = f_{M+1}$ we have
\[
\mathbf r_M = \left(r^{k_1},r^{k_2}\right) = \left(\left(\frac{1}{2^{1/f_M}}\right)^{f_M},\left(\frac{1}{2^{1/f_M}}\right)^{f_{M+1}}\right) = \left(\frac{1}{2},\frac{1}{2^{f_{M+1}/f_M}}\right)
\]
and hence $\bfr_M\to \bfr_\phi$ as $M\to\infty$. In Figure \ref{fig:GoldenComplexDimensionsLatticeApprox}, $\phi$ approximated by $f_{M+1}/f_M$ for $M=2,\ldots,9$. 
\end{example}

The sequence of lattice scaling vectors in the Lemma \ref{lem:SequenceOfLatticeScalingVectors} give rise to a sequence of lattice sets which converge to a given nonlattice set.

\begin{theorem}
\label{thm:LatticeApproximation}
Let $\bfPhi = \{\varphi_j\}_{j=1}^N$ be a nonlattice self-similar system on $\R^m$ with attractor $F$ and scaling vector $\bfr$. Then there exists a sequence of lattice self-similar systems $(\bfPhi_M)_{M=1}^\infty$ with scaling vector $\bfr_M$ and attractor $F_M$ for each $M\in\N$ such that each of the following statements holds as $M\to\infty$:
\begin{enumerate}
	\item $\bfr_M\to\bfr$ componentwise; and
	\item $F_M\to F$ in the Hausdorff metric.
\end{enumerate}	
If, in addition, $\bfPhi$ is $\delta$-disjoint, then 
\begin{itemize}	 
	\item[\emph{(iii)}] for large enough $M$, $\bfPhi_M$ is $\delta_M$-disjoint for some $\delta_M>0$ and $\delta_M\to\delta$.
\end{itemize}
\end{theorem}


\begin{proof}
The key to the construction of the lattice sets $F_M$ lies in replacing the scaling ratios of a given nonlattice system $\bfPhi$ by those of the lattice scaling ratios stemming from an application of Lemma \ref{lem:SequenceOfLatticeScalingVectors} as follows:
\begin{enumerate}
	\item For each $j = 1,\cdots, N$ we have $\varphi_j(\cdot)= r_jQ_j(\cdot)+\mathbf{t}_j$ by Proposition \ref{prop:SelfSimilarCharacterization}.
	\item For each $j=1,\ldots,N$ and each $M\in\N$, define $r_{M,j}$ using the simultaneous Diophantine approximation provided by Lemma \ref{lem:SequenceOfLatticeScalingVectors}.
	\item For each $M\in\N$, define $\bfPhi_M= \{\varphi_{M,j}\}_{j=1}^N$ by 
	\begin{align*}
		\varphi_{M,j}(\cdot)	&:= r_{M,j}Q_j(\cdot)+\mathbf{t}_j	
	\end{align*} 
	for each $j=1,\cdots,N$.	 Also, define $F_M$ to be the attractor of $\bfPhi_M$.
\end{enumerate}
By Theorem 11.1 in Chapter III of \cite{Barns}, which (for our purposes) says that attractors of self-similar systems depend continuously on the scaling vectors, we have $F_M\to F$ since $\bfr_M\to\bfr$ as $M\to\infty$. Therefore, $\displaystyle{(\bfPhi_M)_{M=1}^{\infty}}$ is a sufficient sequence of lattice self-similar systems. 

If $\bfPhi$ is strongly separated, then it readily follows from Definitions \ref{def:HausdorffMetric} and \ref{def:deltadisjoint} that, for large enough $M$, $\bfPhi_M$ is strongly separated. Also, assuming $\delta_M$ does \textit{not} tend to $\delta$ yields a contradiction in light of the fact that $F_M\to F$.
\end{proof}

\section{The Theory of Complex Dimensions of Fractal Strings} 
\label{sec:TheoryComplexDimensionsFractalStrings}

\subsection{Fractal strings, complex dimensions, and zeta functions}
\label{sec:FractalStringsComplexDimensionsZetaFunctions}

This section provides a summary of some of the results from the theory of fractal strings, complex dimensions, and zeta functions, with a focus on Minkowski measurability (in terms of \textit{inner} Minkowski content and dimension) of certain subsets of the real line. See \cite{LapvF6} for a full and thorough introduction to the theory of complex dimensions of fractal strings. These results both motivate and provide a foundation for the new results explored in Sections \ref{sec:BCZFSelfSimilarSets} and \ref{sec:RelatedResultsAndFutureWork} below.

\begin{definition} 
\label{def:FractalString}
An \textit{ordinary fractal string} $\Omega$ is a bounded open subset of the real line. 
A \textit{fractal string} $\calL=(\ell_j)_{j=1}^\infty$ is a nonincreasing sequence of positive real numbers such that $\lim_{j \rightarrow \infty}\ell_j=0$. By a mild abuse of notation, one may think of a fractal string $\calL$ as the multiset 
\begin{align*}
	\calL	&=\{ l_n : l_n \textnormal{ has multiplicity } m_n, n \in \N \},
\end{align*}
where $(l_n)_{n\in\N}$ is the strictly decreasing sequence of distinct lengths $\ell_j=l_n$ and the multiplicity $m_n$ is the number indeces $j$ such that $\ell_j=l_n$.
\end{definition}

\begin{remark}
\label{rmk:OrdinaryFractalStringProperties}
An ordinary fractal string $\Omega$ can be expressed as the countable union of pairwise disjoint open intervals. Throughout this work, as in \cite{LapvF6}, we assume an ordinary fractal string $\Omega$ comprises infinitely many connected components with lengths determining a fractal string $\calL$.
\end{remark}

\begin{definition}
\label{def:GeometricZetaFunction}
Let $\calL=(\ell_j)_{j=1}^\infty$ be a fractal string. The \textit{abscissa of convergence} of the Dirichlet series $\sum_{j=1}^\infty\ell_j^s$ where $s\in\C$ is defined by
\begin{align*}
	\sigma	&:=\inf \left\{ \alpha \in \R : \textstyle{\sum_{j=1}^\infty}\ell_j^\alpha<\infty \right\}.
\end{align*}
The \textit{dimension of a fractal string} $\calL$, denoted by $D_{\calL}$, is defined by $D_{\calL}=\sigma$. The \textit{geometric zeta function} of $\calL$, denoted by $\zeta_\calL$, is defined by
\begin{align*}
	\zeta_\calL(s)	&:= \sum_{j=1}^\infty\ell_j^s= \sum_{n=1}^\infty m_nl_n^s,
\end{align*}
where $s \in \C$ such that $\Real(s)>D_\calL$ and the lengths $l_n$ and multiplicities $m_n$ are as in Definition \ref{def:FractalString}. Note that $\zeta_\calL$ is holomorphic on the half-plane $\Real(s)>D_\calL$ (see \cite[\S VI.2]{Ser}).
\end{definition}

The following theorem is a restatement of Theorem 1.10 of \cite{LapvF6} and justifies the notion of referring to $D_\calL$ as a dimension.

\begin{theorem}
\label{thm:FGCDDimensionsEqual}
Let $\Omega$ be an ordinary fractal string comprising infinitely many connected components such that $\vol^1(\Omega)=\sup\Omega-\inf\Omega$. Then $D_\calL=\overline{\dim}_M(\partial\Omega)$, where $\partial\Omega$ denotes the boundary of $\Omega$.
\end{theorem}

\begin{remark}
\label{rmk:InnerMinkowski}
We are also interested in the \textit{inner Minkowski content, dimension,} and \textit{measurability} of subsets of $\R$.  In Definition \ref{def:MinkDim}, setting $m=1$ and replacing $\vol^m$ with the \textit{inner} volume of an ordinary fractal string $\Omega$, given by
\begin{align*}
	V(\ep)	&:=\vol^1(\Omega\cap(\partial\Omega)_\ep),	
\end{align*}
yields inner Minkowski content, dimension, and measurability (and the upper and lower counterparts) which we denote by ${}_i\scrM, \dim_{iM},$ etc., respectively. In this setting, an analog of Theorem \ref{thm:FGCDDimensionsEqual} holds: If $\Omega$ is an ordinary fractal string comprising infinitely many connected components, then $D_\calL=\overline{\dim}_{iM}(\partial\Omega)$.
\end{remark}

Note that the assumption $\vol^1(\Omega)=\sup\Omega-\inf\Omega$ in Theorem \ref{thm:FGCDDimensionsEqual} ensures that $\overline{\dim}_M(\partial\Omega)=\overline{\dim}_{iM}(\partial\Omega)$. This condition is not always satisfied.

\begin{example}
\label{eg:SVC4}
There are ordinary fractal strings for which $\vol^1(\Omega)\neq\sup\Omega-\inf\Omega$ and the conclusion of Theorem \ref{thm:FGCDDimensionsEqual} does not hold. For instance, complements of certain Smith-Volterra-Cantor sets (some of which are also called \textit{fat Cantor sets}) have this property. To construct a particular example, mimic the ``middle-third removal'' construction of the Cantor set $C$ as follows: Remove an open interval of length $1/4$ centered at $1/2$ from the unit interval $[0,1]$, leaving a compact set comprising the disjoint union of two closed intervals. Repeat the process in a recursive fashion for each positive integer $n\geq 1$ by removing an interval of length $1/4^n$ from the center of the $2^{n-1}$ connected components at that stage, leaving a compact set comprising $2^n$ connected components of equal length. Define $C_4$ to be the intersection of all the (nonempty) compact sets that remain after each step. 

The set $C_4$ is a nonempty, compact, perfect, nowhere dense subset of $[0,1]$ with \textit{positive} 1-dimensional Lebesgue measure. Indeed, the open set $\Omega_4:=[0,1]\backslash C_4$ is an ordinary fractal string with lengths $\calL_4$ given by
\begin{align} \notag
	\calL_4	&= \{1/4^n : 1/4^n \textnormal{ occurs with multiplicity } 2^{n-1}, n\in\N \}.
\end{align}
As such, the length of $C_4$ is given by
\begin{align} \notag
	\vol^1(C_4)	&= 1-\sum_{n=1}^\infty \frac{2^{n-1}}{4^n} = 1-\frac{1/4}{1-(2/4)} = \frac{1}{2}.
\end{align}
Note that $C_4$ cannot be the attractor of a self-similar system. Additionally, we have $\dim_MC_4=1$ and $C_4$ is Minkowski measurable with Minkowski content $\scrM(C_4)=1/2$. 

However, as explained in Example \ref{eg:SVC4ComplexDimensions}, $C_4$ is not \textit{inner} Minkowski measurable with respect to its \textit{inner} Minkowski dimension given by $\dim_{iM}C_4=D_{\calL_4}=1/2$. This fact is closely tied to the structure of the \textit{complex dimensions} of $\calL_4$, as indicated in Theorem \ref{thm:CriterionForMinkowskiMeasurability}.
\end{example}

The complex dimensions of a fractal string are defined as the poles of a meromorphic extension of its geometric zeta function.

\begin{definition}
\label{def:ScreenWindowComplexDimensions}
Let $\calL$ be a fractal string. A \textit{screen} $S$ is a contour 
\begin{align*}
	S:=\{f(t)+it\in\C : t \in \R\},
\end{align*}
where $f:\R\to[-\infty,D_\calL]$ is a continuous function. A \textit{window} $W\subseteq\C$ is the set of points which lie to the right of a screen $S$ given by
\begin{align*}
	W	&:= \{s \in \C : \Real(s) \geq f(\Imag(s))\},
\end{align*}
where $f$ defines $S$. Let $W\subseteq\C$ be a window contained in an open connected neighborhood on which $\zeta_\calL$ has a meromorphic extension but such that $\zeta_\calL$ does not have a pole on the corresponding screen $S$. The set of \textit{visible complex dimensions} of $\calL$ is the set $\calD_\calL(W)$ given by
\begin{align*}
	\calD_\calL(W)	&:= \left\{ \omega \in W : \zeta_\calL \textnormal{ has a pole at } \omega \right\}.
\end{align*}
If $W=\C$, (i.e., if $\zeta_\calL$ has a meromorphic extension to all of $\C$), then
\begin{align*}
	\calD_\calL	&:=\calD_\calL(\C)= \left\{ \omega \in \C : \zeta_\calL \textnormal{ has a pole at } \omega \right\}
\end{align*}
is called the set of \textit{complex dimensions} of $\calL$.

By a mild abuse of notation, both the geometric zeta function of $\calL$ and its meromorphic extension to some window are denoted by $\zeta_\calL$.
\end{definition}

One of the key results in the theory of complex dimensions for ordinary fractal strings is the criterion for Minkowski measurability provided by Theorem \ref{thm:CriterionForMinkowskiMeasurability} below. This result involves the following counting function. 

\begin{definition}
\label{def:GeometricCountingFunction} Let $\calL$ be a fractal string. The \textit{geometric counting function} of $\calL$, denoted by $N_\calL$, if defined for $x>0$ by
\begin{align*}
	\displaystyle N_\calL(x) &:= \#\{j\in \N:\ell_j^{-1}\leq x\}= \sum_{n \in \N, \, l_n^{-1} \leq\, x}m_n.
\end{align*}
\end{definition}

%

Theorem \ref{thm:SimplePoleCondition} is a partial restatement of Theorem 1.17 of \cite{LapvF6}.

\begin{theorem}
\label{thm:SimplePoleCondition}
Let $\calL$ be a fractal string of dimension $D_\calL$ and assume that $\zeta_\calL$ has a meromorphic extension to a neighborhood of $D_\calL$. If
\begin{align*}
	N_\calL(x^{D_\calL})=O(x^{D_\calL}) \quad \textnormal{as } x\to\infty,
\end{align*}
then $\zeta_\calL$ has a simple pole at $D_\calL$.
\end{theorem}

The following lemma is partially a restatement of Lemma 13.110 of \cite{LapvF6}, and a more general result is Proposition \ref{prop:D=sigma} below. A \textit{much} stronger result, which is beyond the scope of this paper, holds in the setting of Mellin transforms and generalized fractal strings (viewed as measures) as described in \cite{LapvF6}. 

\begin{lemma}
\label{lem:GeometricZetaFunctionIntegralTransformCountingFunction}
Let $\calL$ be a fractal string. Then 
\begin{align*}
	\zeta_\calL(s)	&= s\int_0^\infty N_\calL(x)x^{-s-1}dx
\end{align*}
and, moreover, the integral converges (and hence, the equation 
holds)  if and only if $\sum_{j=1}^\infty \ell_j^s$ converges, i.e., if and only if $\Real(s)>D_\calL=\sigma$.
\end{lemma}

In \cite[\S 5.3]{LapvF6}, the terms \textit{languid} and \textit{strongly languid} describe the growth of a geometric zeta function $\zeta_\calL$ in terms of three technical conditions called \textbf{L1}, \textbf{L2}, and \textbf{L2}$'$. \textbf{L1} is a polynomial growth condition along horizontal lines (in the complex plane) necessarily avoiding the poles of $\zeta_\calL$,  \textbf{L2} is a polynomial growth condition along the vertical direction of a corresponding screen, and \textbf{L2}$'$ is a stronger version of \textbf{L2}. A fractal string $\calL$ is \textit{languid} if $\zeta_\calL$ satisfies \textbf{L1} and \textbf{L2}, and is \textit{strongly languid} if  $\zeta_\calL$ satisfies \textbf{L1} and \textbf{L2}$'$. These conditions allow for some of the key results in \cite{LapvF6} to hold, such as Theorem 8.15 therein, which appears here as Theorem \ref{thm:CriterionForMinkowskiMeasurability}. This theorem is a primary motivation behind the new results presented in Section \ref{sec:BCZFSelfSimilarSets} and especially Corollary \ref{cor:CriterionBCMeasurability} and Proposition \ref{prop:NonlatticeBCContentResidue}. (See also \cite{LapRaZu14,LapRoZu,Morales,Sar}.)

\begin{theorem}
\label{thm:CriterionForMinkowskiMeasurability}
Let $\Omega$ be an ordinary fractal string, with lengths $\calL$, comprising infinitely connected components such that $\vol^1(\Omega)=\sup\Omega-\inf\Omega$. Suppose $\calL$ is languid for a screen passing between the vertical line $\Real(s)=D_\calL$ and all of the complex dimensions of $\calL$ with real part strictly less than $D_\calL$ and not passing through zero. Then the following are equivalent:
	
\begin{enumerate}[{\normalfont (i)}]
	\item $D_\calL$ is the only complex dimension with real part $D_\calL$, and it is simple.
	\item $N_\calL(x)=E\cdot x^{D_\calL} + o(x^{D_\calL})$ for some positive constant $E$.
	\item $\partial\Omega$, the boundary of $\Omega$, is Minkowski measurable.		
\end{enumerate}
	Moreover, if any of these conditions is satisfied, then 
\begin{align}
\label{eqn:MinkowskiContentAsResidue}
	\scrM(\partial\Omega) &= 2^{1-D_\calL}\frac{E}{1-D_\calL}=2^{1-D_\calL}\frac{\res(\zeta_\calL(s);D_\calL)}{D_\calL(1-D_\calL)}.
\end{align}
\end{theorem}

\begin{remark}
\label{rmk:DropLengthHypothesis}
In Theorem \ref{thm:CriterionForMinkowskiMeasurability}, if the hypothesis  $\vol^1(\Omega)=\sup\Omega-\inf\Omega$ is dropped, the results still hold but with \textit{inner} Minkowski measurability in part (iii) and  \eqref{eqn:MinkowskiContentAsResidue} gives the \textit{inner} Minkowski content ${}_i\scrM(\partial\Omega)$. 
\end{remark}

\begin{example}
\label{eg:SVC4ComplexDimensions}
By Example \ref{eg:SVC4}, we have
\begin{align*}
	\zeta_{\calL_4}(s)	&= \sum_{n=1}^\infty 2^{n-1}4^{-ns} = \frac{4^{-s}}{1-2\cdot4^{-s}},
\end{align*}
for $\Real(s)>D_{\calL_4}=1/2$. The closed form on the right-hand side of this equation allows for a meromorphic extension of $\zeta_{\calL_4}$ to all of $\C$ and it is used to verify that $\calL_4$ is strongly languid. It follows that
\begin{align*}
	\calD_{\calL_4}	&:=\calD_{\calL_4}(\C)=\left\{ \frac{1}{2} + i\frac{2\pi}{\log 4}z : z \in \Z \right\}.
\end{align*}
Note that $D_{\calL_4}=1/2$ is \textit{not} the only complex dimension with real part equal to $1/2$, so by Theorem \ref{thm:CriterionForMinkowskiMeasurability}, the fat Cantor set $C_4$ is not inner Minkowski measurable (even though it is Minkowski measurable, see Example \ref{eg:SVC4}). 
\end{example}

\begin{example}
\label{eg:CantorString}
The Cantor string is the ordinary fractal string given by $\Omega_{CS}:=[0,1] \backslash C$ where $C$ is the Cantor set. The corresponding fractal string $\calL_{CS}$ (also referred to as the Cantor string) is given by
\begin{align*}
	\calL_{CS}	&= \{ 1/3^n : 1/3^n \textnormal{ has multiplicity } 2^{n-1}, n \in \N \}.
\end{align*}
The geometric zeta function of the Cantor string, denoted by $\zeta_{CS}$, is given by
\begin{align*}
	\zeta_{CS}(s)	&:= \zeta_{\calL_{CS}}(s)= \sum_{n=1}^\infty 2^{n-1}3^{-ns} = \frac{3^{-s}}{1-2\cdot3^{-s}},
\end{align*} 
for $\Real(s)>D_{\calL_CS}=\dim_MC=\log_32$. The closed form on the right-hand side of this equation allows for a meromorphic extension of $\zeta_{CS}$ to all of $\C$ and it is used to show that $\calL_{CS}$ is strongly languid. It follows that the set of complex dimensions of the Cantor string is given by
\begin{align*}
	\calD_{CS}	&:=\calD_{\calL_{CS}}(\C)=\left\{ \log_3{2} + i\frac{2\pi}{\log 3}z : z \in \Z \right\}.
\end{align*}
We have that $D_{CS}:=D_{\calL_{CS}}=\log_{3}2=\dim_BC$. Moreover, $D_{CS}$ is \textit{not} the only complex dimension with real part equal to $D_{CS}$, so by Theorem \ref{thm:CriterionForMinkowskiMeasurability}, the Cantor set $C$ is not Minkowski measurable. This fact was established in \cite{LapPo1} via the equivalence of (ii) and (iii) and showing that (ii) does not hold. Actually, in \cite{LapPo1}, $\scrM^*$ and $\scrM_*$ are explicitly computed and shown to be different (with $0<\scrM_*<\scrM^*<\infty$). 
\end{example}

Note that in part (i) of Theorem \ref{thm:CriterionForMinkowskiMeasurability}, the only complex dimensions of interest are the ones which have real part equal to $D_\calL$. This motivates the following definition, which agrees in spirit with the definition of \textit{principal complex dimensions} of \cite{LapRaZu14}.

\begin{definition}
\label{def:PrincipalComplexDimensions}
Let $\calL$ be a fractal string with dimension $D_\calL$ and visible complex dimensions $\calD_\calL(W)$ associated with a window $W$. The \textit{principal complex dimensions} of $\calL$, denoted by $\dim_{PC}\calL$, is given by
\begin{align*}
	\dim_{PC}\calL	&:=\{\omega\in\calD_\calL(W)\subseteq\C:\Real(\omega)=D_\calL\}.
\end{align*}
\end{definition}

\begin{remark}
\label{rmk:PCDsWindow}
Note that the principal complex dimensions $\dim_{PC}\calL$ are independent of the choice of window $W$. For the Cantor string $\calL_{CS}$ and the fractal string $\calL_4$ we have $\dim_{PC}\calL_{CS}=\calD_{CS}$ and $\dim_{PC}\calL_4=\calD_{\calL_4}$, respectively.
\end{remark}

For an ordinary fractal string $\Omega$ with lengths $\calL$, the complex dimensions of $\calL$ provide more than just a criterion for the Minkowski measurability of the boundary $\partial\Omega$. In particular, if $\calL$ is strongly languid, the geometric counting function $N_\calL$ can be written as a sum over the complex dimensions of the residues of $\zeta_\calL$ as in the following theorem, a special case of Theorem 5.14 in \cite{LapvF6}.

\begin{theorem}
\label{thm:CountingOverComplex}
Let $\calL$ be a strongly languid fractal string. Then, for some $a>0$ and all $x>a$, the pointwise explicit formula for $N_\calL$ is given by
\begin{align*}
	N_\calL(x) &= \sum_{\omega \in \calD_\calL(W)}
\textnormal{res}\left(\frac{x^{s}\zeta_\calL(s)}{s};\omega\right) + \{\zeta_\calL(0)\},
\end{align*}
where the term in braces is included only if $0 \in W\backslash \calD_\calL(W)$. If, in addition, all of the principal complex dimensions of $\calL$ are simple, then for $x>a$ we have
\begin{align}
\label{eqn:CountingOverSimpleComplex}
	N_\calL(x) &= \sum_{\omega \in \dim_{PC}\calL}\frac{x^\omega}{\omega}\textnormal{res}\left(\zeta_\calL(s);\omega\right)  +o(x^{D_\calL}).
\end{align}
\end{theorem}

\begin{figure}
	\centering
	\includegraphics[scale=0.26]{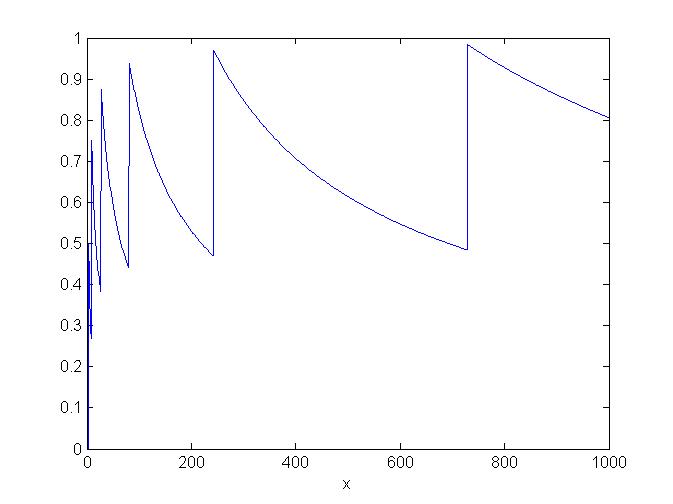} 
	\includegraphics[scale=0.26]{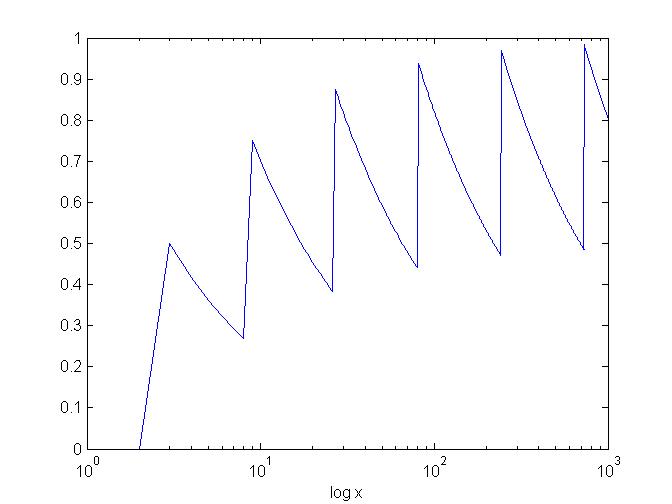}
	\caption{The geometric oscillations of the Cantor string $\Omega_{CS}$ seen in the plot and semilog plot of $N_{CS}(x)/x^{D}$, on the left and right, respectively. Here, $N_{CS}$ is the geometric counting function of the Cantor string and $D=\dim_BC=\log_32$. See Example \ref{eg:CantorStringCounting}. (The function is discontinuous, the vertical line segments are an artifact of the program used to generate the images.)}\label{fig:ContNcs}
\end{figure}

\begin{example}
\label{eg:CantorStringCounting}
As determined in \cite[p.23]{LapvF6}, the geometric counting function of the Cantor string $\Omega_{CS}$, denoted by $N_{CS}$, is given by 
\begin{align}
\label{eqn:CantorStringCounting}
	N_{CS}(x)	&=2^{n}-1=\frac{1}{2\log 3}\sum_{k\in\Z}\frac{x^{D+ikp}}{D+ikp}-1,
\end{align}
where $D=D_{CS}=\log_32$, $p=2\pi/\log3$, and $n=[\log_3x]$ where $[y]$ denotes the integer part of $y\in\R$. Note that in \cite{LapvF6}, this formula is derived directly using a particular Fourier series. Nonetheless, the formula for $N_{CS}$ provided in \eqref{eqn:CantorStringCounting} also follows from an application of Theorem \ref{thm:CountingOverComplex} since each complex dimension $\omega=D+ikp\in\calD_{CS}$ is simple and the residue of $\zeta_{CS}$ at each $\omega$ is independent of $k\in\Z$. The common value of these residues is given by
\begin{align*}
	\res(\zeta_{CS}(s);D+ikp)	&=\frac{1}{2\log 3}.
\end{align*}

Now, consider Figure \ref{fig:ContNcs}. The (nearly) log-periodic structure of $N_{CS}(x)/x^{D}$ can be seen in this figure, which is indicated by the \textit{geometric oscillations of order} $D$ inherent to $\Omega_{CS}$. Also, the fact that the Cantor set $C$ is \textit{not} Minkowski measurable (as discussed above in Example \ref{eg:CantorString}) can be inferred from this figure. In this case, Minkowski content and inner Minkowski content coincide. (See Remark \ref{rmk:InnerMinkowski}.) Since $\vol^1(C)=0$,  $\scrM^*(C)= {}_i\scrM^*(C)$ and $\scrM_*(C)= {}_i\scrM_*(C)$. As shown in \cite{LapPo1}, we have 
\begin{align*}
	\scrM^*(C)={}_i\scrM^*(C)	&=2^{2-D}\approx 2.5830, \quad \textnormal{and}\\
	\scrM_*(C)= {}_i\scrM_*(C)	&=2^{1-D}D^{1-D}(1-D)^{1-D}\approx 2.4950.
\end{align*}
In light of Theorem \ref{thm:CountingOverComplex} and the fact that 
\begin{align*}
	V(\ep):=\vol^1(\Omega\cap(\partial\Omega)_\ep)	&= 2\ep\cdot N_\calL(1/2\ep)+\sum_{j:\ell_j<2\ep}\ell_j,
\end{align*}
\eqref{eqn:CantorStringCounting} indicates that the complex dimensions in $\calD_{CS}$ encode the geometric oscillations which force $C$ to fail to be Minkowski measurable.

Similar results hold for $\Omega_4$ and the fat Cantor set $C_4$ in the context of inner Minkowski content, dimension, and measurability. (See Examples \ref{eg:SVC4} and \ref{eg:SVC4ComplexDimensions}.)
\end{example}

\subsection{Lattice/nonlattice dichotomy of self-similar strings}
\label{sec:Lattice/NonlatticeDichotomyOfSelfSimilarStrings}

Many results regarding the special case of the lattice/nonlattice dichotomy for (nontrivial) self-similar subsets of the real line have been established. See \cite{LapvF6}.

\begin{definition}
\label{def:SelfSimilarString}
Let $\bfPhi=\{\varphi_j\}_{j=1}^N$ be a self-similar system on $\R$ that satisfies the open set condition with attractor $F$ and scaling vector $\bfr=(r_j)^N_{j=1}$ where $N\geq 2$. Let $I=[\inf F,\sup F]$ and $L=\sup F-\inf F$. If $\sum_{j=1}^Nr_j<1$, then the nonempty ordinary fractal string $\Omega=I\backslash F$ is called a \textit{self-similar string}. If $\bfPhi$ is lattice (or nonlattice), then $\Omega$ is a \textit{lattice (or nonlattice) string}. 
\end{definition}

\begin{remark}
\label{rmk:Gaps}
Let $\Omega=I\backslash F$ be a self-similar string. Let $K$ denote the positive number of connected components in $I\backslash(\cup_j\varphi_j(I))$ which have positive length, and let $g_kL$ denote the length of the $k$th  connected component for $k=1,\ldots,K$ arranged so that $0<g_1\leq\cdots\leq g_K<1$ and $\sum_{k=1}^Kg_k+\sum_{j=1}^Nr_j =1$. The $g_k$ are called the \textit{gaps} of $\Omega$. (See \cite[Chapter 2]{LapvF6}.) Note that $\partial\Omega=F$ in this case.
\end{remark}

\begin{theorem}
\label{thm:GeometricZetaFunctionSelfSimilarString}
Let $\Omega$ be a self-similar string (as in Definition \ref{def:SelfSimilarString}) with lengths $\calL$. Then the geometric zeta function $\zeta_\calL$ has a meromorphic extension to the whole complex plane given by 
\begin{align*}
	\zeta_\calL(s)	&=\frac{L^s\textstyle{\sum_{k=1}^K}g_k^s}{1-\textstyle{\sum_{j=1}^N}r_j^s}, \qquad s\in\C.
\end{align*}
Here, $\zeta_\calL(1)=L$ is the total length of $\Omega$ as well as the length of the interval $I$.
\end{theorem}

\begin{corollary}
\label{cor:ComplexDimensionsSelfSimilarString}
Let $\Omega$ be a self-similar string (as in Definition \ref{def:SelfSimilarString}) with lengths $\calL$. Then the set of complex dimensions $\calD_\calL$ is a subset of $\calS_\bfr$, the complex solutions of the Moran equation \eqref{eqn:Moran} given by \eqref{eqn:ComplexMoranSolutions}.
Also, each complex dimension has a multiplicity at most that of the corresponding solution. If, in addition, $\Omega$ has a single value for the gaps $g_1=\cdots =g_K$, then $\calD_\calL=\calS_\bfr$.
\end{corollary}

\begin{example}
\label{eg:CantorGoldenComplexDimensions}
The Cantor string $\Omega_{CS}$ and the \textit{Golden string} $\Omega_\phi:=[0,1]\backslash A_\phi$ are self-similar strings, each with a single gap. So, Corollary \ref{cor:ComplexDimensionsSelfSimilarString} applies to $\Omega_{CS}$ and $\Omega_\phi$, and the complex dimensions are given by $\calS_{\bfr_C}$ and $\calS_{\bfr_\phi}$, accordingly. The set of complex dimensions $\calD_{CS}$ of the Cantor string $\Omega_{CS}$ (or $\calL_{CS}$) is determined in Example \ref{eg:CantorString}. The complex dimensions $\calD_{\calL_\phi}$ of the Golden string $\Omega_\phi$ are the solutions of the transcendental equation
\begin{align}
\label{eqn:GoldenTranscendental}
	2^{-\omega}+2^{-\phi\omega}=1, \quad \omega\in\C.
\end{align}
See Figure \ref{fig:GoldenComplexDimensionsLatticeApprox} for images of successive approximations of the complex dimensions of the Golden string. These images were not obtained through solving \eqref{eqn:GoldenTranscendental} directly but rather through the approximation of the complex dimensions, stated in terms of the structure of roots of Dirichlet polynomials, as detailed in Chapter 3 of \cite{LapvF6} and described heuristically in Remark \ref{rmk:LatticeRootsApprox}.

Note that $\Omega_4$ is not a self-similar string since $C_4$ is not a self-similar set. Nonetheless, the lengths $\calL_4$ \textit{are} the lengths of some self-similar string.
\end{example}

The complex dimensions of self-similar strings and the box-counting complex dimensions of many self-similar subsets of some Euclidean space (see Section \ref{sec:BCZFSelfSimilarSets}) are often given by the set of complex solutions of Moran equations of the form \eqref{eqn:Moran}. These sets are denoted by $\calS_\bfr$ and defined in \eqref{eqn:ComplexMoranSolutions}. 

The following theorem is a small part of Theorems 3.6 and 3.23 in \cite{LapvF6} which provides a wealth of information regarding the structure of the set $\calS_\bfr$.

\begin{theorem}
\label{thm:StructureOfComplexDimensions}
Let $\bfr=(r_1,\ldots,r_N)$ be a scaling vector. If $\bfr$ is lattice, then the elements of $\calS_\bfr$	are obtained by finding the complex solutions $z$ of the equation
\begin{align*}
	\sum_{j=1}^Nr_j^\omega	&=\sum_{=1}^Mm_uz^{k_u}=1,
\end{align*}
where $z=e^{-\omega\log{r^{-1}}}$, $m_u$ is the number of $j$ such that $r_j=r^{k_u}$, and $M$ is the number of distinct values among the $r_j$. Hence there exist finitely many solutions $\omega_1,\ldots,\omega_q$ such that 
\begin{align*}
	\calS_\bfr	&=\left\{\omega_t+inp : n\in\Z, t=1,\ldots,q\right\}, 
\end{align*}
where $p=2\pi/\log{r^{-1}}$. 

If $\bfr$ is nonlattice, then $\omega=D_\bfr$ is the only element of $\calS_\bfr$ with real part equal to $D_\bfr$ and all others have real part less than $D_\bfr$. Also, there exists a sequence of elements of $\calS_\bfr$ approaching $\Real(s)=D_\bfr$ from the left.
\end{theorem}

The Minkowski measurability of the boundary of a self-similar string is directly related to whether the string is lattice or nonlattice. The following theorem is a combination of the results stated in Theorem 8.23 and 8.36 from \cite{LapvF6}.

\begin{theorem}
\label{thm:LatticeNonlatticeStringMeasurability}
The boundary of a self-similar string is Minkowski measurable if and only it is nonlattice.
\end{theorem}

\begin{example}
\label{example:LatticeNonLatticeStrings}
The Golden string is nonlattice, and as such Theorem \ref{thm:LatticeNonlatticeStringMeasurability} implies that $A_\phi$ in Minkowski measurable. Thus, Theorem \ref{thm:StructureOfComplexDimensions} implies that set of principal complex dimensions of the Golden string is a singleton comprising $D_{\calL_\phi}\approx .77921$. That is, $\dim_{PC}\calL_{\phi}=\{D_{\calL_\phi}\}$. Moreover, $D_\phi$ is a simple pole of $\zeta_{\calL_\phi}$, so Theorem \ref{thm:CriterionForMinkowskiMeasurability} applies and the (inner) Minkowski content of $A_\phi$ is given by \eqref{eqn:MinkowskiContentAsResidue}.

\end{example}

\begin{figure}
\begin{center}
\includegraphics[scale=.40]{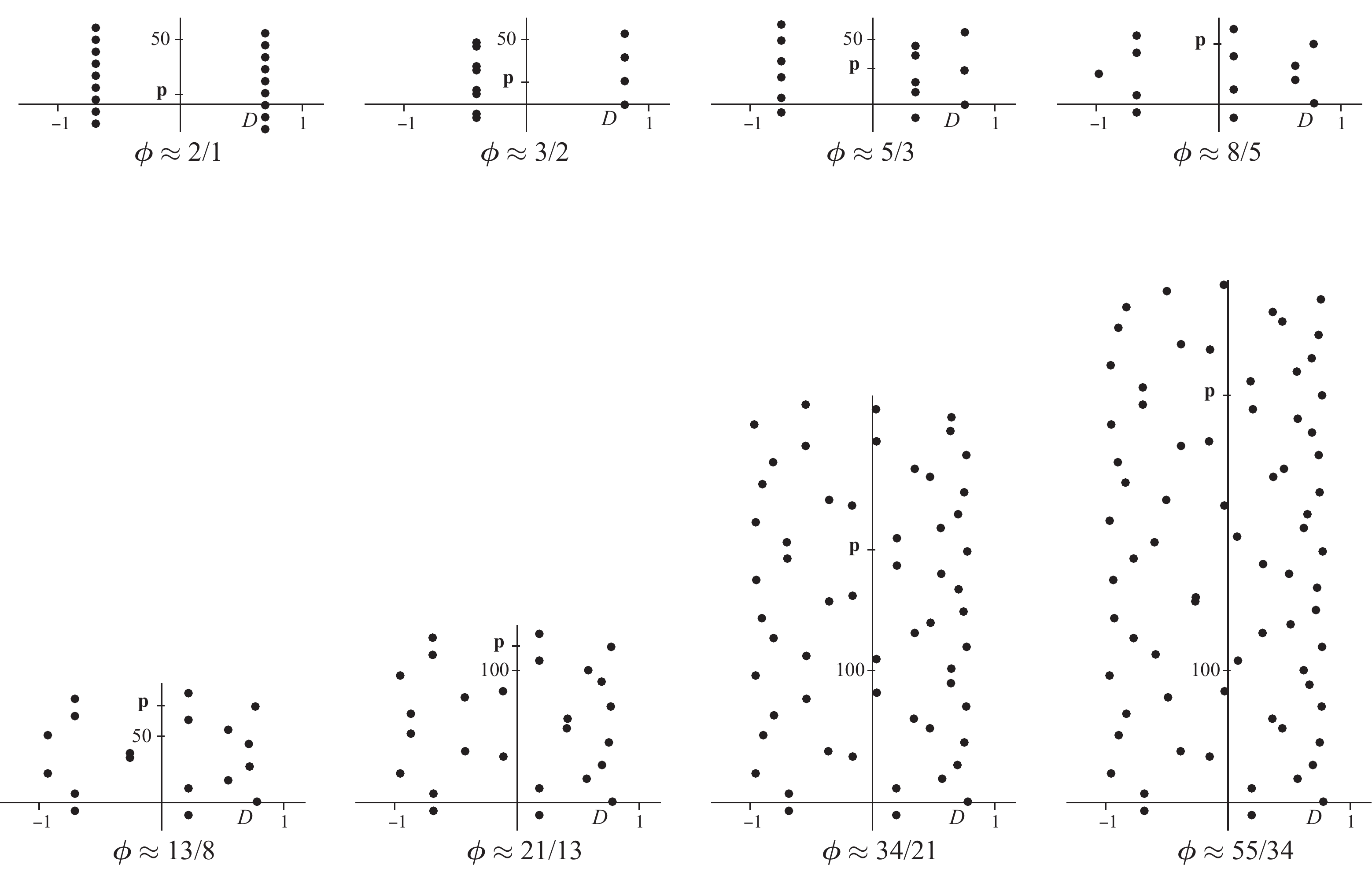}
\end{center}
\caption{A lattice approximation of $\calD_{\calL_\phi}$, the complex dimensions of the Golden string $\Omega_\phi$ with lengths $\calL_\phi$. The plots show the complex dimensions $\calD_M=\{z \in\C: 2^{-z} + 2^{-z\phi_M} = 1\}$ where $\phi_M=f_{M+1}/f_M$ approximates $\phi$ for $M=2,\ldots,9$. In each case, the point $D$ denotes the Minkowski dimension of the approximating attractor, and the figure repeats with period $\bfp$. Note, however, that $\calD_{\calL_\phi}$ itself is \textit{not} periodic. See Examples \ref{eg:GoldenStringSystem}, \ref{eg:GoldenString}, and \ref{eg:CantorGoldenComplexDimensions} as well  as Theorem \ref{thm:StructureOfComplexDimensions} and Remark \ref{rmk:LatticeRootsApprox}.}
\label{fig:GoldenComplexDimensionsLatticeApprox}
\end{figure}

\begin{remark}
\label{rmk:LatticeRootsApprox}
Chapter 3 of \cite{LapvF6} provides a thorough description of the manner in which the set of roots of a \textit{nonlattice} Dirichlet polynomial are approximated by the set of roots of a \textit{lattice} Dirichlet polynomial. In this paper, this approximation is discussed in terms of sequences with convergence denoted by $\calS_{\bfr_M}\to\calS_\bfr$ as $M\to\infty$ and loosely described as follows: Given a nonlattice scaling vector $\bfr$ and any fixed $T>0$, there is a lattice scaling vector $\bfr_M$ (constructed through Lemma \ref{lem:SequenceOfLatticeScalingVectors}) such that each of the roots in $\calS_\bfr$ with imaginary part less than $T$ (in absolute value) is approximated in a uniform manner by a root in $\calS_{\bfr_M}$ and the multiplicity of the corresponding roots coincide. Moreover, the oscillatory period $\bfp$ (i.e., the period in the imaginary direction) of the roots in $\calS_{\bfr_M}$ with maximal real part $D$ is much smaller than $T$. See Figure \ref{fig:GoldenComplexDimensionsLatticeApprox} for a collection of images which show the approximation of the roots in $\calS_{\bfr_\phi}$ associated with the nonlattice scaling vector $\bfr_\phi$. By Corollary \ref{cor:ComplexDimensionsSelfSimilarString}, we have $\calS_{\bfr_\phi}=\calD_{\calL_\phi}$.

The approximations seen in Figure \ref{fig:GoldenComplexDimensionsLatticeApprox} are given by a sequence of roots stemming from lattice scaling vectors whose components depend directly on the ratios of Fibonacci numbers, see Example \ref{eg:GoldenString}. Note that, by Corollary \ref{cor:ComplexDimensionsSelfSimilarString}, the convergence $\calS_{\bfr_M}\to\calS_\bfr$ as $M\to\infty$ described above also describes the ``quasiperiodic'' behavior of the structure of the complex dimensions of a self-similar string.
\end{remark}

As noted above, the approximation used in the proof of Theorem \ref{thm:LatticeApproximation} makes explicit use of Lemma \ref{lem:SequenceOfLatticeScalingVectors} which, in turn, brings along the convergence of complex dimensions in the context of self-similar strings as described in Remark \ref{rmk:LatticeRootsApprox} and Figure \ref{fig:GoldenComplexDimensionsLatticeApprox}. This begs the question as to whether the convergence of complex dimensions also holds in suitable context for self-similar subsets of \textit{any} Euclidean space and not just subsets of $\R$ (as is the case for self-similar strings). A potentially suitable context is described in the remainder of the paper by making use of the results presented in this section. In particular, one of the goals of studying \textit{box-counting zeta functions} in Sections \ref{sec:BCZFSelfSimilarSets} and \ref{sec:RelatedResultsAndFutureWork} is to find a framework in which the lattice/nonlattice dichotomy can be discussed in terms \textit{box-counting complex dimensions}.

\section{Box-Counting Zeta Functions of Self-Similar Sets}
\label{sec:BCZFSelfSimilarSets}

\subsection{Box-counting fractal strings and zeta functions}
\label{BoxCountingFractalStringsZetaFunctions}
The material presented in this section follows from the results of Lalley in \cite{Lal88} as discussed in Sections \ref{sec:LatticeNonlatticeDichotomyMeasurability} and \ref{sec:SimultaneousDiophantineApproximation}, along with those determined by the third author in \cite{Sar}. The work done in \cite{Sar} was motivated by that of Lapidus and van Frankenhuisjen in \cite{LapvF6} (as outlined in Section \ref{sec:TheoryComplexDimensionsFractalStrings}) and the \textit{box-counting fractal strings and zeta functions} introduced by Lapidus, ~\v Zubrini\'c, and the fifth author and in \cite{LapRoZu}.

\begin{definition}
\label{def:BCFString}
Let $A\subseteq \R^m$ be a bounded infinite set and let $N_B(A,\cdot)$ denote the box-counting function of $A$ given in Definition \ref{def:BoxCountingFunction}. Let the range of $N_B(A,\cdot)$ be denoted by $(M_n)_{n\in \N}$, a strictly increasing sequence of positive integers. For each $n\in\N$, let $l_n$ be the scale given by
\begin{align*}
	l_n^{-1}	&=\sup\{x \in (0,\infty): N_B(A,x)=M_n\}.
\end{align*}
That is, $l_n^{-1}$ is the positive real number where $N_B$ jumps from $M_n$ to $M_{n+1}$. Now, let $m_1=M_2$, and let $m_n=M_{n+1}-M_n$, for $n\geq 2$. The \textit{box-counting fractal string} of $A$, denoted by $\calL_B$, is the fractal string with distinct lengths $(l_n)_{n\in \N}$ and corresponding multiplicities $(m_n)_{n\in \N}$.
\end{definition}

The following technical proposition and lemma were introduced in \cite{LapRoZu} as part of the development of a theory of box-counting fractal strings, zeta functions, and complex dimensions. In particular, they allow one to make use of the results of Lapidus and van Frankenhuijsen in \cite{LapvF6} (as outlined in Section \ref{sec:TheoryComplexDimensionsFractalStrings}) in the context of box-counting functions, dimension, and content.

\begin{proposition}
\label{prop:BCFSstructure}
Let $A\subseteq \R^m$ be a bounded infinite set. For each $n\in\N$, let $x_n=l_n^{-1}$ where $l_n$ is as defined in Definition \ref{def:BCFString}. Then $\lim_{n \rightarrow \infty}x_n=\infty$. Also $x_1>0$, $N_B(A,x_1)=1$, the union of $N_B(A,x_n)$ is the range of $N_B(A,\cdot)$, and for $x_{n-1}<x\leq x_n$ we have that $N_B(A,x_{n-1})<N_B(A,x)=N_B(A,x_n)$. 
\end{proposition}

\begin{lemma}
\label{lem:BCFScountingfunction}
For a bounded infinite set $A\subseteq \R^m$ and any $x \in (x_1,\infty)\backslash (x_n)_{n\in \N}$,
\begin{align*}
	N_B(A,x)	&=N_{\calL_B}(x).
\end{align*}
Also, for $x\in (0,x_1]$, $N_B(A,x)=1$ while $N_{\calL_B}(x)=0$.
\end{lemma}

The following is one of the key results of \cite{LapRoZu} and, in part, motivates the definition of \textit{box-counting zeta function} given just below.

\begin{theorem}
Let $A \subseteq \R^m$ be a bounded infinite set. Then 
\begin{align*}
	D_{\calL_B}	&=\overline{\dim}_BA.
\end{align*}
\end{theorem}

\begin{definition}
\label{def:BCZF}
The \textit{box-counting zeta function} of bounded infinite set $A\subseteq \R^m$, denoted by $\zeta_B$, is the geometric zeta function of $\calL_B$. That is,
\begin{align*}
	\zeta_B(s)	&:=\zeta_{\calL_B}(s) = \sum_{n=1}^\infty m_nl_n^s,
\end{align*}
where $\Real(s)>D_{\calL_B}=\overline{\dim}_BA$. The \textit{(box-counting) complex dimensions} of $A$ on a suitably defined window, denoted by $\calD_B(W)$, is given by $\calD_B(W):=\calD_{\calL_B}(W)$. If $W=\C$, then $\calD_B(\C)$ is denoted by $\calD_B$.
\end{definition}


\begin{example}
\label{eg:BCFSofUnitInterval}
Consider the unit interval $[0,1]$. Then $N_B([0,1],x)=[x/2]+1$ and hence we have $l_n=x_n^{-1}=2n$, $m_1=2$, and $m_n=1$ for each $n\geq 2$. So, 
\begin{align*}
	\zeta_B(s)	&=\frac{2}{2^s}+\sum_{n=2}^\infty\frac{1}{(2n)^s}=\frac{1}{2^s}+\frac{1}{2^s}\zeta(s),
\end{align*}
where $\Real(s)>1$ and $\zeta(s)$ denotes the Riemann zeta function. It is well-known that $\zeta$ converges for $\Real(s)>1$ and has a simple pole at $s=1$, so $D_{\calL_B}=\dim_B[0,1]=1$. 
\end{example}

\subsection{Self-similar sets under separation conditions}
\label{sec:SelfSimilarSetsSeparation}
This section focuses on some the results of \cite{Sar} which, in part, follow from results of \cite{Lal88}. All of the results presented here make use of either the (strong) open set condition or strongly separated self-similar systems.

\begin{lemma}
\label{lem:BCFdelta}
Let $\bfPhi$ be a $\delta$-disjoint self-similar system with attractor $F$ and scaling vector $\bfr=(r_j)_{j=1}^N$. If $x>\delta^{-1}$, then $N_B(F,x)	=\sum^N_{j=1}N_B(F,r_jx).$ Moreover, for any $x>0$,
\begin{align}
\label{eqn:LalleyRenewalEquation}
	N_B(F,x)	&=\sum^N_{j=1}N_B(F,r_jx)+L(x)
\end{align}
where $L(x)$ is a nonpositive, nondecreasing integer valued step function with a finite number of steps that is bounded below by $1-N$ and vanishes for $x>\delta^{-1}$.
\end{lemma}

\begin{proof}
A ball with radius $x^{-1}\leq \delta$ and center in $\varphi_j(F)$ does not intersect any other $\varphi_k(F)$ for $k\neq j$. This fact, combined with the self-similarity of $F$, imply that 
\begin{align*}
	N_B(F,x)	&= N_B\left(\bfPhi(F),x\right)=\sum^N_{j=1}N_B(\varphi_j(F),x)=\sum^N_{j=1}N_B(F,r_jx).
\end{align*} 
The properties of $L(x)$ follow readily from Lemma \ref{prop:BCFSstructure}.
\end{proof}

The following theorem provides a (nearly) closed form for the box-counting zeta function of a self-similar set whose counting function satisfies a renewal equation of the form \eqref{eqn:LalleyRenewalEquation}. 

\begin{theorem}
\label{thm:BCZF}
Let $\bfPhi$ be a self-similar system that satisfies the open set condition with attractor $F$ and scaling vector $\bfr=(r_j)_{j=1}^N$. Let $\calL_B$ be the box-counting fractal string of $F$. Suppose $N_B(F,x)=\sum^N_{j=1}N_B(F,r_jx)+L(x)$ where $s\int_{x_1}^{\infty}L(x)x^{-s-1}dx$ converges for $\Real(s)>D_\bfr=\dim_BF$. Then the box-counting zeta function of $F$ is given by
\begin{align}
\label{eqn:ClosedBCZFwOSC}
	\zeta_B(s)	&=(x_1^{-1})^s+\frac{(N-1)(x_1^{-1})^s+E(s)}{1-\textstyle{\sum\limits}_{j=1}^Nr_j^s}
\end{align}
where $\Real(s)>D_\bfr=\dim_BF$ and $E(s):=s\int_{x_1}^{\infty}L(x)x^{-s-1}dx$.
\end{theorem}

\begin{proof}
Temporarily let $s\in\R$ such that $s>D_\bfr=\dim_BF$. By Lemma \ref{lem:BCFScountingfunction}, $N_B(F,x)=N_{\calL_B}(x)$ for $x\in (x_1,\infty) \backslash (x_n)_{n\in \N}$. Since $N_B(F,x) \neq N_{\calL_B}(x)$ for an at most countable number of values $x\geq x_1$ and $N_{\calL_{\calB}}(x)=0$ for $x<x_1$, applying Lemma \ref{lem:GeometricZetaFunctionIntegralTransformCountingFunction} along with \eqref{eqn:LalleyRenewalEquation} yields
\begin{align}
\label{eqn:zetaB1}
	\zeta_B(s)	&=\sum^N_{j=1}s\int_{x_1}^\infty N_B(F,r_jx)x^{-s-1}dx+s\int_{x_1}^{\infty}L(x)x^{-s-1}dx.
\end{align}
Let $E(s)=s\int_{x_1}^{\infty}L(x)x^{-s-1}dx$ and apply the substitution $u=r_jx$ to get
\begin{align}
\label{eqn:zetaB2}
	\sum^N_{j=1}s\int_{x_1}^\infty N_B(F,r_jx)x^{-s-1}dx	&=\sum^N_{j=1}r_j^ss\int_{r_jx_1}^\infty N_B(F,u)u^{-s-1}du.
\end{align}
Since $0<r_j<1$ implies $r_jx_1<x_1$ and we have $N_B(F,u)=1$ for $u<x_1$. Also,
\begin{align}
	s\int_{r_jx_1}^\infty N_B(F,u)u^{-s-1}du	&=s\int_{r_jx_1}^{x_1}u^{-s-1}du+s\int_{x_1}^\infty N_B(F,u)u^{-s-1}du \notag \\
&=(x_1^{-1}r_j^{-1})^s-(x_1^{-1})^s+\zeta_B(s) \label{eqn:zetaB3}.
\end{align}
By combining \eqref{eqn:zetaB1}, \eqref{eqn:zetaB2}, and \eqref{eqn:zetaB3} we get
\begin{align*}
	\zeta_B(s)	&=\sum_{j=1}^Nr_j^s\left((x_1^{-1}r_j^{-1})^s-(x_1^{-1})^s+\zeta_B(s)\right)+E(s)\\
	&=(x_1^{-1})^s\sum_{j=1}^N(1-r_j^s)+\zeta_B(s)\sum_{j=1}^Nr_j^s+E(s).
\end{align*}
Solving for $\zeta_B(s)$ and simplifying yields \eqref{eqn:ClosedBCZFwOSC}. Hence, the Principal of Analytic Continuation allows $\zeta_B$ to extend so as to be holomorphic on the half-plane $\Real(s)>D_\bfr=\dim_BF$ (see \cite[\S VI.2]{Ser}).
\end{proof}

The following corollary is a key step in the development of a lattice/nonlattice dichotomy from the perspective of the theory of (box-counting) complex dimensions associated with strongly separated self-similar systems. 

\begin{corollary}
\label{cor:BCZFdeltaDisjoint}
Let $\bfPhi$ be a $\delta$-disjoint self-similar system with attractor $F$ and scaling vector $\bfr=(r_j)_{j=1}^N$, and let $\calL_B$ be the box-counting fractal string of $F$. Then the box-counting zeta function of $F$ has a closed form given by
\begin{align}
\label{eqn:ClosedDeltaDisjointBCZF}
	\zeta_B(s)	&=({x_1}^{-1})^s+\frac{(N+e_m-1)({x_1}^{-1})^s+\textstyle{\sum\limits}_{k=m}^{n-1}(e_{k+1}-e_k)({y_k}^{-1})^s-e_n\delta^s}{1-\textstyle{\sum\limits}_{j=1}^{N}{r_j}^s},
\end{align}
where $\Real(s)>D_\bfr=\dim_BF$ and the values of $m,n$, the $e_k$ and the $y_k$ satisfy the following properties: $m,n\in\N$ with $m\leq n$;  with $1-N=e_1<\cdots<e_n<0$ with $e_k\in\Z$; $0<y_1<\cdots<y_n\leq \delta^{-1}$; and $m \in \N$ is the smallest number such that $y_m>x_1$.
\end{corollary}

The proof is omitted but can be found in \cite{Sar}. Essentially, it follows from a careful decomposition of renewal equation of \eqref{eqn:LalleyRenewalEquation} and the evaluation of the integral defining $E(s)$ in Theorem \ref{thm:BCZF}. The values of $e_k$ and $y_k$ are determined by the structure of $L(x)$, which is a step function for $0<x\leq\delta^{-1}$, as provided by Lemma \ref{lem:BCFdelta}.

Corollary \ref{cor:BCcdim} provides a couple of results regarding the complex dimensions of a strongly separated self-similar set. In this corollary and Theorem \ref{thm:LatticeApproximationAdditional}, the numerator on the right-hand side of \eqref{eqn:ClosedDeltaDisjointBCZF} plays a role in determining the structure of these complex dimensions of strongly separated self-similar sets. So, given a $\delta$-disjoint self-similar set $F$, let $h$ denote the numerator of \eqref{eqn:ClosedDeltaDisjointBCZF} given by
\begin{align}
\label{eqn:Numerator}
	h(s)	&:=(N+e_m-1)({x_1}^{-1})^s+\sum\limits_{k=m}^{n-1}(e_{k+1}-e_k)({y_k}^{-1})^s-e_n\delta^s,
\end{align}
where $s\in\C$ has large enough real part. 

\begin{corollary}
\label{cor:BCcdim}
Let $\bfPhi$ be a strongly separated self-similar system with attractor $F$ and scaling vector $\bfr=(r_j)_{j=1}^N$. Then
\begin{align*}
	\calD_B	&=\calD_{\calL_B}\subseteq \calS_\bfr =\left\lbrace s \in \C : 1-\textstyle{\sum_{j=1}^{N}}{r_j}^s = 0 \right\rbrace.
\end{align*}
Additionally,  equality holds if we have $h(s)=0$ if and only if $s\notin\calS_\bfr$.
\end{corollary}

\begin{remark}
\label{rmk:}
Note that in the strongly separated case, the box-counting fractal string $\calL_B$ of a self-similar set is strongly languid. This fact follows from an argument similar to that made for geometric zeta functions of self-similar strings in \cite[\S 6.4]{LapvF6}. Hence, Theorem \ref{thm:CountingOverComplex} applies and yields a formula for the box-counting function $N_B(A,x)$ over the box-counting complex dimensions $\calD_B$ for large enough $x$.
\end{remark}

Under certain conditions, the box-counting complex dimensions of a strongly separated nonlattice set are approximated by those of lattice sets.

\begin{theorem}
\label{thm:LatticeApproximationAdditional}
Suppose $\bfPhi$ is a $\delta$-disjoint nonlattice self-similar system on $\R^m$ with attractor $F$, scaling vector $\bfr$, and box-counting complex dimensions $\calD$. Additionally, suppose $h(s)=0$ if and only if $s\notin\calS_\bfr$. Then there exists a sequence of lattice self-similar systems $(\bfPhi_M)_{M=1}^\infty$ with scaling vectors $\bfr_M$, attractors $F_M$, and box-counting complex dimensions $\calD_M$ for each $M\in\N$, respectively, such that each of the following holds as $M\to\infty$:
\begin{enumerate}
	\item $\bfr_M\to\bfr$ componentwise;
	\item $F_M\to F$ in the Hausdorff metric;
	\item for large enough $M$, $F_M$ is $\delta_M$-disjoint for some $\delta_M>0$ and $\delta_M\to\delta$; and
	\item $\calD=\calS_\bfr$ and $\calS_{\bfr_M}\to\calD$ in the sense described in Remark \ref{rmk:LatticeRootsApprox} or, more specifically, in Chapter 3 of \cite{LapvF6}.
\end{enumerate}
Let $h_M$ denote the numerator determined by $F_M$ and $\delta_M$ given by \eqref{eqn:Numerator}, accordingly. If, in addition to the hypotheses above, we have $h_M(s)=0$ if and only if $s\notin\calS_{\bfr_M}$ for all $M\geq M_0$ where $M_0$ is some positive integer, then 
\begin{itemize}
	\item[\emph{(v)}] $\calD_M=\calS_{\bfr_M}$ for $M\geq M_0$ and $\calD_M\to\calD$ in the sense described in Remark \ref{rmk:LatticeRootsApprox}.
\end{itemize}
\end{theorem}

\begin{proof}
Theorem \ref{thm:LatticeApproximation} immediately yields parts (i), (ii), and (iii). Parts (iv) and (v) follow immediately from the application of Corollary \ref{cor:BCcdim} to $\bfPhi$ and $\bfPhi_M$ for large enough $M$, accordingly.
\end{proof}

Similar but more limited results hold for a self-similar set that satisfies the open set condition (and, equivalently by Theorem \ref{thm:OSCimpliesSOSC}, the \textit{strong} open set condition). The following restatement of Proposition 1 of \cite{Lal88} follows from renewal theory. 

\begin{proposition}
\label{prop:SOSCerror}
Let $\bfPhi$ be a self-similar system that satisfies the open set condition with attractor $F$ and scaling vector $\bfr=(r_j)_{j=1}^N$. Then, with $D=\dim_BF=D_\bfr$, we have $N_B(F,x)=\sum^N_{j=1}N_B(F,r_jx)+L(x)$ where 
\begin{align*}
	|L(x)|\leq\gamma x^{D-\ep}
\end{align*}
for some constants $\gamma, \ep>0$.
\end{proposition}

The next proposition makes use of the previous one to show that self-similar sets satisfying the open set condition have box-counting zeta functions of a certain form. It is the result of personal communication between Lapidus and the fifth author. 

\begin{proposition}
\label{prop:LalleyErrorControl}
Let $\bfPhi$ be a self-similar system that satisfies the open set condition with attractor $F$ and scaling vector $\bfr=(r_j)_{j=1}^N$. Also, let $\calL_B$ be the box-counting fractal string of $F$ and $D:=\dim_BF=D_\bfr$. Then there exists some $\ep>0$ such that $E(s):=s\int_{x_1}^{\infty}L(x)x^{-s-1}dx$ converges for $\Real(s)>D-\ep$.
\end{proposition}

\begin{proof}
By Proposition \ref{prop:SOSCerror}, we have $N_B(F,x)=\sum^N_{j=1}N_B(F,r_jx)+L(x)$ where $|L(x)|\leq\gamma x^{D-\ep}$ for some constants $\gamma, \ep>0$. Now, temporarily let $s$ denote a real number such that $s>D-\ep$. Since $D-\ep-s<0$, we have $\frac{- \gamma s}{D-\ep-s}>0$. Therefore, since $\lim_{a \rightarrow \infty}a^{D-\ep-s}=0$, we have
\begin{align*}
|E(s)| 
&\leq  s\int_{x_1}^{\infty}\gamma x^{D-\ep} x^{-s-1}dx 
\leq \lim_{a \rightarrow \infty} \gamma s\int_{x_1}^{a} x^{D-\ep-s-1}dx
= \frac{-\gamma s x_1^{D-\ep-s}}{D-\ep-s}.
\end{align*}
Thus, by allowing $s\in\C$ we have that $E(s)$ converges for $\Real(s)>D-\ep$.
\end{proof}

Under the hypotheses of Proposition \ref{prop:LalleyErrorControl} and Theorem \ref{thm:BCZF}, the \textit{principal} complex dimensions are solutions of the corresponding Moran equation. 

\begin{corollary}
\label{cor:SOSCComplexDimensions}
Let $\bfPhi$ be a self-similar system that satisfies the open set condition with attractor $F$ and scaling vector $\bfr=(r_j)_{j=1}^N$. Then $\dim_{PC}\calL_B\subseteq \calS_\bfr$.
\end{corollary}

\begin{proof}
Since $\bfPhi$ satisfies the open set condition, by Proposition \ref{prop:LalleyErrorControl} there exists some $\ep>0$ such that $E(s)=s\int_{x_1}^{\infty}L(x)x^{-s-1}dx$ converges for $\Real(s)>D_\bfr-\ep$ where $D_\bfr$. So, by Theorem \ref{thm:BCZF} the poles of $\zeta_B$ with real part equal to $D_\bfr$  (i.e., principal complex dimensions) must be roots of the Dirichlet polynomial $1-\textstyle{\sum_{j=1}^{N}}{r_j}^s$. 
\end{proof}

\subsection{Examples of box-counting zeta functions of self-similar sets}
\label{sec:ExamplesBCZF}
Theorem \ref{thm:BCZF} and Corollary \ref{cor:BCZFdeltaDisjoint} allow one to calculate the box-counting zeta functions of the Sierpi\'{n}ski gasket $S_G$ and the totally disconnected 1-dimensional set $F_1$, as done here. See Examples \ref{eg:SierpinskiGasket}, \ref{eg:4byQuarter}, and \ref{eg:SimilarityDimensions} above, as well as Examples \ref{eg:BCZF4byQuarter} and \ref{eg:BCZFGasket} in this section. In both cases, $\bfPhi_1$ and $\bfPhi_S$ satisfy the (strong) open set condition, so Theorem \ref{thm:BCZF}, Proposition \ref{prop:LalleyErrorControl}, and Corollary \ref{cor:SOSCComplexDimensions} apply, accordingly. 

\begin{example}
\label{eg:BCZF4byQuarter}
The lattice set $F_1$ from Example \ref{eg:4byQuarter} stems from the $\delta$-disjoint self-similar system $\bfPhi_1$ (with $\delta=1/2)$). The box-counting function of $F_1$ is determined in \cite{LapRoZu}. For $0<x\leq 2$, we have
\[
N_B(F_1,x)=
\begin{cases}
1, & 0 < x \leq 2/\sqrt{2}, \\
2, & 2/\sqrt{2} < x \leq 8/\sqrt{17}, \\
3, & 8/\sqrt{17} < x \leq 2.
\end{cases}
\]
Hence, $M_1=1, M_2=2$, and $M_3=3$. See Figure \ref{fig:figure}. 
For $x>2$ and $n\in\N$ we have
\begin{align}
\label{eqn:BCF4byQuarter}
	N_B(F_1,x)	&=
\begin{cases}
4^n, & 2\cdot4^{n-1} < x \leq 2/\sqrt{2}\cdot4^n, \\
2\cdot4^n, & 2/\sqrt{2}\cdot4^n < x \leq 8/\sqrt{17}\cdot4^n, \\
3\cdot 4^n, & 8/\sqrt{17}\cdot4^n < x \leq 2\cdot4^n.
\end{cases}
\end{align}

It follows that $F_1$ is not box-counting measurable. We have that $\dim_BF_1=1$ by Moran's Theorem (Theorem \ref{thm:MoransTheorem}). So, considering that the extreme behavior of $N_B(F_1,x)/x$ occurs at the endpoints of the intervals defined in \eqref{eqn:BCF4byQuarter}, we also have
\begin{align*}
	\scrB^*(F_1)	&=\limsup_{x\to\infty}\frac{N_B(F_1,x)}{x^1}=\lim_{n\to\infty}\frac{4^n}{2\cdot 4^{n-1}}=2, \qquad \textnormal{and}\\
	\scrB_*(F_1)	&=\liminf_{x\to\infty}\frac{N_B(F_1,x)}{x^1}=\lim_{n\to\infty}\frac{4^n}{(2/\sqrt{2})\cdot 4^n}=\frac{\sqrt{2}}{2}.
\end{align*}
Hence, $F_1$ exhibits geometric oscillations of order $\dim_BF_1=1$.

Since $\bfPhi_1$ is $\delta$-disjoint with $\delta=1/2$, Lemma \ref{lem:BCFdelta} implies that the box-counting function of $F_1$ satisfies the renewal equation
\begin{align*}
	N_B(F_1,x)	&=4N_B(F_1,x/4)+L(x).
\end{align*}
Moreover, $L(x)$ is explicitly given by
\[
L(x)=
\begin{cases}
-3, & 0 < x \leq 2/\sqrt{2}, \\
-2, & 2/\sqrt{2} < x \leq 8/\sqrt{17}, \\
-1, & 8/\sqrt{17} <x \leq 2,\\
0, & x > 2.
\end{cases}
\]
So, by Corollary \ref{cor:BCZFdeltaDisjoint} (and in agreement with the results of \cite{LapRoZu}),
\begin{align*}
	\zeta_B(s) 
	&=(\sqrt{2}/2)^s+\frac{(\sqrt{2}/2)^s+(\sqrt{17}/8)^s+(1/2)^s}{1-4\cdot4^{-s}}.
\end{align*}
It follows that the box-counting complex dimensions of $F_1$ (and the principal complex dimensions) are given by
\begin{align*}
	\calD_{\calL_B}	&=\dim_{PC}\calL_B=\calS_{\bfr_1}=\left\lbrace1+i\frac{2\pi}{\log4}k : k \in \Z\right\rbrace.
\end{align*}
Additionally, each of these complex dimensions is a simple pole, so Theorem \ref{thm:CountingOverComplex} applies and a pointwise explicit formula for $N_B(F_1,x)$ is given by \eqref{eqn:CountingOverSimpleComplex}.
\end{example}

\begin{example}
\label{eg:BCZFGasket}
The Sierpi\'{n}ski gasket $S_G$ is the attractor of the lattice self-similar system $\bfPhi_S$. This system satisfies the open set condition but is not strongly separated (see Examples \ref{eg:SierpinskiGasket} and \ref{eg:OSCAndSOSC}). The box-counting function of $S_G$ is given by
\begin{align}
\label{eqn:BCFGasket}
N_B(S_G,x)=
\begin{cases}
1, & 0 < x \leq 2, \\
\displaystyle\frac{3^n+3}{2}, & 2^n < x \leq 2^{n+1} \quad \textnormal{for}\quad n \geq 1.
\end{cases}
\end{align}

As in the case of $F_1$, the Sierpi\'{n}ski gasket $S_G$ is not box-counting measurable. We have that $\dim_BS_G=\log_23=:D$. So, considering that the extreme behavior of $N_B(S_G,x)/x^D$ occurs at the endpoints of the intervals defined in \eqref{eqn:BCFGasket}, we also have
\begin{align*}
	\scrB^*(S_G)	&=\limsup_{x\to\infty}\frac{N_B(S_G,x)}{x^D}=\lim_{n\to\infty}\frac{(3^n+3)/2}{2^{n\log_23}}=\frac{1}{2}, \qquad \textnormal{and}\\
	\scrB_*(S_G)	&=\liminf_{x\to\infty}\frac{N_B(S_G,x)}{x^D}=\lim_{n\to\infty}\frac{(3^n+3)/2}{2^{(n+1)\log_23}}=\frac{1}{6}.
\end{align*}
Hence, $S_G$ exhibits geometric oscillations of order $\dim_BS_G=\log_23$.

Note that the box-counting function of $S_G$ satisfies the renewal equation 
\begin{align*}
	N_B(S_G,x)	&=3N_B(S_G,x/2)+L(x).
\end{align*}
Observe that for $0<x\leq 2$, $L(x)=-2$, and for $2<x\leq 4$, $L(x)=0$. Now suppose that $2^n<x\leq 2^{n+1}$ for $n\geq 2$. Then $N(S_G,x)=(3^n+3)/2$. Since $2^{n-1}<x/2\leq  2^n$, $N(S_G,x/2)=(3^{n-1}+3)/2$. So, for $x>4$ we have  that
\begin{align*}
	L(x)&=N_B(S_G,x)-3N_B(S_G,x/2) =-3.
\end{align*}
Applying Theorem \ref{thm:BCZF} and evaluating $E(s)=s\int_{x_1}^{\infty}L(x)x^{-s-1}dx=-3(1/4)^s$ yields
\begin{align*}
	\zeta_B(s)&=(1/2)^s+\frac{2(1/2)^s-3(1/4)^s}{1-3\cdot 2^{-s}},
\end{align*}
for $\Real(s)>\dim_BS_G=\textnormal{log}_23$. Note that if $\omega \in \C$ is a root of $1-3\cdot 2^{-s}$, then $(1/2)^\omega=1/3$ and  the numerator of $\zeta_B$ equals $2(1/3)-3(1/3)^2=1/3$. Thus, there is no cancellation of the solutions of the Moran equation $1-3\cdot 2^{-s}=0$. Hence, the principal (box-counting) complex dimensions of the Sierpi\'{n}ski gasket are given by
\begin{align*}
	\dim_{PC}\calL_B=\calS_{\bfr_S}=\left\{\log_23+i\frac{2\pi}{\log2}k : k \in \Z\right\}.
\end{align*}
Additionally, each of the principal complex dimensions is a simple pole, so Theorem \ref{thm:CountingOverComplex} applies and a pointwise explicit formula for $N_B(S_G,x)$ is given by \eqref{eqn:CountingOverSimpleComplex}.
\end{example}

\section{Related Results and Future Work}
\label{sec:RelatedResultsAndFutureWork}

\subsection{Generalized content and zeta function}
\label{sec:GenContentAndZetaFunction}
This section provides a preliminary framework for the study of complex dimensions in the setting of a simple class of \textit{Dirichlet type integrals} (see \cite{LapRaZu14} and references therein) as developed in \cite{Morales}. Note that in the setting developed here there is no underlying geometry. 

\begin{definition}
\label{def:GeneralStuff}
Let $f:(0,\infty)\rightarrow[0,\infty)$. The \textit{exponent of} $f$ is defined by
	\begin{align*}
	\alpha_f		&:=\inf\{\alpha\geq 0:f(x)=O(x^{\alpha}) \hspace{.10in}\text{ as } x\rightarrow\infty\}.
	\end{align*}
Let $f:(0,\infty)\rightarrow[0,\infty)$ and assume there exists an $x_o>0$ such that $f(x)=0$ for all $x\in(0,x_o)$. Then the \textit{dimension} of $f$, $D_f$, is defined by 
	\begin{align*}
		D_f	&:=\inf\left\{t\in\R:t\int_{0}^{\infty}f(x)x^{-t-1}dx<\infty \right\}.
	\end{align*}
Additionally, the \textit{zeta function} of $f$, denoted by $\zeta_f$, is given by 
	\begin{align*}
		\zeta_f(s)	&:=s\int_{0}^{\infty}f(x) x^{-s-1}dx,
	\end{align*}
for $s\in\C$ such that $\Real(s)>D_f$. Let $W\subseteq\C$ be a window containing an open connected neighborhood on which $\zeta_f$ has a meromorphic extension. The set of \textit{(visible) complex dimensions} of $f$ is defined as
	\begin{align*}
		\calD_f(W)	&:=\{w\in W:\zeta_f \text{ has a pole at } w\}.
	\end{align*}
	If $\zeta_f$ has a meromorphic extension to all of $\C$, then $\calD_f:=\calD_f(\C)$ denotes the set of \textit{complex dimensions} of $f$.
\end{definition}

%

The following theorem is essentially an analog of Lemma 13.110 of \cite{LapvF6} (Lemma \ref{lem:GeometricZetaFunctionIntegralTransformCountingFunction} in this paper), Theorem 13.111 of \cite{LapvF6}, and Lemma 3.13 of \cite{LapRoZu}. 

\begin{proposition}
\label{prop:D=sigma}
Let $f:(0,\infty)\rightarrow[0,\infty)$. Suppose there exists an $x_o>0$ such that $f(x)=0$ for all $x\in(0,x_o)$ and suppose $f$ is nondecreasing with $\lim_{x\to\infty}f(x)=\infty$. Then $\alpha _f=D_f$.
\end{proposition}

\begin{definition}
	\label{def:steady}
	Let $f:(0,\infty)\rightarrow[0,\infty)$ be nondecreasing. The \textit{upper} and \textit{lower content} of $f$ are defined, respectively, by 
	\begin{align*}
	\scrC^* &:=\limsup_{x\rightarrow\infty}\frac{f(x)}{x^{D_{f}}}, \quad \textnormal{and} \quad \scrC_*: =\liminf_{x\rightarrow\infty}\frac{f(x)}{x^{D_{f}}}.
	\end{align*}
	If $\scrC^*=\scrC_*$, then the \textit{content} of $f$, denoted by $\scrC$, is defined to be this common value. Additionally, if $0<\scrC^*=\scrC_*<\infty$, then $f$ is said to be \textit{steady}.
\end{definition}

The following theorem is a type of generalization of Theorem \ref{thm:CriterionForMinkowskiMeasurability} stated in terms of steady functions, but without any geometric context such as Minkowski or box-counting measurability. The proof is omitted since it follows that of Theorem 8.15 in \cite{LapvF6}, mutatis mutandis.

\begin{theorem}
	\label{thm:SteadyCharacterization}
	Let $f$ be such that $\zeta_f$ is languid for a screen passing between the vertical line $\Real(s)=D_f$ and all the complex dimensions of $f$ with real part strictly less than $D_f$ and not passing through zero. Then the following are equivalent:
	
	\begin{enumerate}[{\normalfont (i)}]
		\item $D_f$ is the only complex dimension with real part $D_f$, and it is simple.
		\item $f(x)=E\cdot x^{D_f} + o(x^{D_f})$ for some positive constant $E$.
		\item $f$ is \textit{steady}.
	\end{enumerate}
\end{theorem}

Theorems \ref{thm:CountingOverComplex} and \ref{thm:SteadyCharacterization} combine to yield the following corollary. 

\begin{corollary}
\label{cor:CriterionBCMeasurability}
Let $A$ be a subset of $\R^m$ such that its box-counting zeta function $\zeta_B$ is strongly languid for a screen passing between the vertical line $\Real(s)=D_{\calL_B}$ and all the (box-counting) complex dimensions of $A$ with real part strictly less than $D_{\calL_B}$ and not passing through zero. Then the following are equivalent:	
\begin{enumerate}[{\normalfont (i)}]
	\item $D_{\calL_B}$ is the only complex dimension with real part $D_{\calL_B}$, and it is simple.
	\item $N_B(A,x)=\scrB\cdot x^{D_{\calL_B}} + o(x^{D_{\calL_B}})$ for some positive constant $\scrB$.
	\item $A$ is box-counting measurable with box-counting content $\scrB$.
\end{enumerate}
If \emph{(i), (ii),} or \emph{(iii)} holds, then 
\begin{align}
\label{eqn:BCContentAsResidue}
	\scrB	&=\scrB(A)=\frac{\res(\zeta_B(A,s);D_{\calL_B})}{D_{\calL_B}}.	
\end{align}
\end{corollary}

\begin{remark}
Note that the unit interval $[0,1]$ has a box-counting zeta function that  satisfies the hypotheses of Corollary \ref{cor:CriterionBCMeasurability} applies, but the Sierpi\'{n}ski gasket $S_G$ and the self-similar set $F_1$ do not. As in Example \ref{eg:BCFSofUnitInterval}, we have
\begin{align*}
	\zeta_B(s)	&=\frac{1}{2^s}+\frac{1}{2^s}\zeta(s),
\end{align*}
where $\Real(s)>1$ and $\zeta(s)$ denotes the Riemann zeta function. It is well-known that $\zeta$ has a simple pole at $s=1$ with $\res(\zeta(s);1)=1$ and its meromorphic extension to $\C$ does not have any other pole with positive real part. Hence, \eqref{eqn:BCContentAsResidue} yields
\begin{align*}	
	\scrB([0,1])	&=\frac{\res(\zeta_B(s);1)}{1}=\frac{1}{2},
\end{align*}
which is in agreement with Example \ref{eg:UnitIntervalMeasurable}.  Also see Examples \ref{eg:UnitIntervalLatticeNonlattice}, \ref{eg:BCZF4byQuarter}, and \ref{eg:BCZFGasket}.
\end{remark}

Proposititon \ref{prop:NonlatticeBCContentResidue} closes the paper and adds to the lattice/nonlattice dichotomy of self-similar sets from the perspective of box-counting measurability.

\begin{proposition}
\label{prop:NonlatticeBCContentResidue}
Suppose $\bfPhi$ is a nonlattice strongly separated self-similar system with scaling vector $\bfr$, attractor $F$, and $\dim_BF=D_\bfr$. Further, suppose $h(s)=0$ if and only if $s\notin\calS_\bfr$. Then $F$ is box-counting measurable and 
\begin{align}
\label{eqn:BCContentNonlattice}
	\scrB(F)	&=\frac{h(D_\bfr)}{D_\bfr\sum_{j=1}^Nr_j^{D_\bfr}\log{r_j^{-1}}}, 	
\end{align}
where $h(s)$ is given by \eqref{eqn:Numerator}.
\end{proposition}

\begin{proof}
By Theorem \ref{thm:LalleyDichotomy}(a), $F$ is box-counting measurable since $\bfPhi$ is nonlattice.

%

Also, Theorem \ref{thm:SimplePoleCondition} applies and we have that $D_\bfr$ is a simple pole of $\zeta_B$. Since  $D_\bfr\in\calS_\bfr$, we have $h(D_\bfr)\neq0$ and the residue of $\zeta_B$ at $D_\bfr$ is readily given by
\begin{align}
\label{eqn:BCNonlatticeResidue}
	\res(\zeta_B;D_\bfr)	&=\frac{h(D_\bfr)}{\sum_{j=1}^Nr_j^{D_\bfr}\log{r_j^{-1}}}.
\end{align}
Finally, Corollary \ref{cor:CriterionBCMeasurability} combines with \eqref{eqn:BCNonlatticeResidue} to yield \eqref{eqn:BCContentNonlattice}.
\end{proof}

\section*{Acknowledgements} The first, second, and fifth authors would like to thank Dr.~Helena Noronha for the support provided to them in 2013 and 2014 through the California State University Alliance for Preparing Undergraduates through Mentoring toward PhDs (PUMP) program via NSF grant DMS--1247679. The authors would also like to thank Dr.~Erin Pearse for the use of the image in Figure \ref{fig:GoldenComplexDimensionsLatticeApprox}.


\medskip
Received xxxx 20xx; revised xxxx 20xx.
\medskip

\end{document}